\documentclass[preprint,12pt]{elsarticle}

\usepackage{amssymb}
\usepackage{amsthm}
\usepackage{amsmath}
\usepackage{mathtools}
\usepackage{extarrows}
\usepackage{paralist}
\usepackage{doi}
\usepackage{stmaryrd}
\usepackage{bussproofs}
\newcommand{\TernaryInfC}{\TrinaryInfC}
\usepackage{tikz}
\usetikzlibrary{automata, positioning, shapes.misc}

\journal{}

\newtheorem{theorem}{Theorem}
\newtheorem{definition}[theorem]{Definition}
\newtheorem{lemma}[theorem]{Lemma}
\newtheorem{corollary}[theorem]{Corollary}
\newtheorem{example}[theorem]{Example}
\newtheorem{observation}[theorem]{Observation}

\newcommand{\Q}{\mathbb{Q}}
\newcommand{\N}{\mathbb{N}}

\newcommand{\eps}{\varepsilon}

\newcommand{\expr}{\mathbb{E}}

\newcommand{\ITE}{\mathsf{ITE}}

\newcommand{\val}{\mathsf{val}}
\renewcommand{\O}{\mathcal{O}}

\newcommand{\pto}{\rightharpoonup}

\newcommand{\kstar}{^{\textstyle *}}
\newcommand{\kostar}{^{\odot\hspace{-0.47em}*}}
\newcommand{\kplus}{^+}

\newcommand{\tg}[1]{\texttt{#1}}
\newcommand{\G}{\mathsf{G}}
\newcommand{\edge}{\mathsf{edge}}
\newcommand{\sspc}{\hspace{0.2em}}

\newcommand{\MSO}{\mathsf{MSO}}
\newcommand{\Word}{\mathsf{Word}}
\newcommand{\tsum}{{\textstyle\sum}}
\newcommand{\tbigwedge}{{\textstyle\bigwedge}}
\newcommand{\tbigvee}{{\textstyle\bigvee}}
\newcommand{\tbigcup}{{\textstyle\bigcup}}
\newcommand{\tto}{\to\hspace{-1.9ex}\to}
\newcommand{\dom}{\mathrm{dom}}
\newcommand{\sem}[1]{\llbracket #1 \rrbracket}
\newcommand{\true}{\mathrm{true}}
\newcommand{\false}{\mathrm{false}}
\newcommand{\first}{\mathsf{first}}
\newcommand{\id}{\mathrm{id}}
\newcommand{\aut}[1]{\mathcal{#1}}

\newcommand{\sepa}{;\sspc}

\newcommand{\rate}{\mathsf{R}}

\newcommand{\op}{\mathit{op}}

\newcommand{\fold}{\mathsf{fold}}

\newcommand{\elseQ}{\mathbin{\texttt{else}}}

\newcommand{\splitQ}{\texttt{split}}

\newcommand{\iterQ}{\texttt{iter}}

\newcommand{\prefsumQ}{\texttt{prefix-sum}}

\newcommand{\SR}{\mathsf{SR}}
\newcommand{\SLR}{\mathsf{SLR}}

\newcommand{\unary}{\mathsf{unary}}

\newcommand{\sink}{\mathsf{sink}}
\newcommand{\R}{\mathcal{R}}
\newcommand{\compl}{{\sim}}
\newcommand{\At}{\mathsf{At}}

\newcommand{\sring}[1]{\mathcal{#1}}

\newcommand{\imax}{i_{\textsf{max}}}

\newcommand{\weight}{\mathsf{weight}}
\newcommand{\last}{\mathsf{last}}

\begin{document}

\begin{frontmatter}


\author[up]{Rajeev Alur}
\author[bgu]{Dana Fisman}
\author[ru]{Konstantinos Mamouras}
\author[usc]{Mukund Raghothaman}
\author[up]{Caleb Stanford}

\address[up]{University of Pennsylvania, Philadelphia, PA, USA}
\address[bgu]{Ben-Gurion University, Be'er Sheva, Israel}
\address[ru]{Rice University, Houston, TX, USA}
\address[usc]{University of Southern California, Los Angeles, CA, USA}

\title{Streamable Regular Transductions}


\address{}

\begin{abstract}
Motivated by real-time monitoring and data processing applications, we develop a formal theory of quantitative queries for streaming data that can be evaluated efficiently. We consider the model of unambiguous Cost Register Automata (CRAs), which are machines that combine finite-state control (for identifying regular patterns) with a finite set of data registers (for computing numerical aggregates). The definition of CRAs is parameterized by the collection of numerical operations that can be applied to the registers. These machines give rise to the class of \emph{streamable regular transductions} ($\SR$), and to the class of \emph{streamable linear regular transductions} ($\SLR$) when the register updates are \emph{copyless}, i.e.\ every register appears at most once in the right-hand-side expressions of the updates. We give a logical characterization of the class $\SR$ (resp., $\SLR$) using MSO-definable transformations from strings to DAGs (resp., trees) without backward edges. Additionally, we establish that the two classes $\SR$ and $\SLR$ are closed under operations that are relevant for designing query languages. Finally, we study the relationship with weighted automata (WA), and show that CRAs over a suitably chosen set of operations correspond to WA, thus establishing that WA are a special case of CRAs.
\end{abstract}


\end{frontmatter}


\section{Introduction}
\label{sec:intro}

Finite-state automata, and the associated class of regular languages, have many appealing properties.
In particular, the class of regular languages has robust, and well understood, expressiveness with multiple
characterizations, and is closed under a variety of language operations.
Furthermore, a finite-state automaton processes an input string in a single left-to-right pass 
using only constant space, in compliance with the {\em streaming\/} model of computation.
This forms the basis of practical applications of automata in design and implementation
of query languages for pattern matching in strings.

In a diverse range of applications such as financial tickers, network traffic monitoring, and click-streams of
web usage, the core computational problem is to map a stream of data items to a numerical value.
For example, suppose the monitoring software at a network router wants
to compute the average number of packets per VoIP (Voice over IP) session 
in the incoming stream of packets.
This requires detection of VoIP sessions, which can be characterized by regular patterns in the input stream, and computing numerical aggregates in a hierarchical manner, where
the cost of a VoIP session is the number of packets it contains, and the cost of a stream is the average of costs
of VoIP sessions it contains.
Recent work on {\em Quantitative Regular Expressions\/} (QREs) \cite{AFR2016QRE} provides a declarative query language
for specifying such computation in a modular fashion, and has led to prototype implementations StreamQRE \cite{MRAIK2017SQRE, AM2017SQRE}
and NetQRE~\cite{YLMMAL2017NQRE}.
These quantitative queries involve a mix of regular patterns and numerical aggregates, and crucially, all queries can be evaluated in a streaming manner. In this paper we develop an understanding of the class of streaming transductions definable by such queries: the streamable regular transductions.

A natural starting point for this study is the model of {\em cost register automata\/} (CRA)
and the associated class of {\em regular transductions\/}~\cite{AdADRY2013CRA}.
A CRA processes a data word, that is, a sequence over $\Sigma\times D$, where $\Sigma$ is a finite set of tags
and $D$ is a set of data values, and outputs a value in $D$ that is computed using a given set of
operations over $D$ (for example, $D$ can be the set of natural numbers, and the set of operations can include
addition, minimum, and maximum).
A CRA has a finite-state control that is updated based only on the tag values of the input word,
and a finite set of (write-only) registers that are updated at each step using the given operations.
The partial functions from $(\Sigma\times D)\kstar$ to $D$ 
computed by the CRAs are called regular with
respect to the allowed set of operations.
The computation of a regular transduction can be viewed as, first mapping a data word to a tree whose
nodes are labeled with the allowed operations, and then evaluating the resulting term,
where the former mapping is an MSO-definable word-to-tree transformation. 
This gives an alternative logical characterization of the class of regular transductions, and allows us to
draw upon the well developed theory of regular tree transformations~\cite{BE2000, EM1999, AdA2017STT}.

The original CRA model is a deterministic machine that processes its input in a single pass.
Its registers can hold data values as well as functions represented by terms with parameters,
and the update at each step is required to be {\em copyless\/}, that is, a register can appear at most once
in the right-hand-side expressions of the updates.
The CRA model we study in section~2 differs in three ways.
First, since our focus is on streamability, and terms can grow linearly with the size of the input stream,
we require registers to hold only values. 
Second, we allow the finite-state control to be updated nondeterministically based on the input tags, but it
is required to be {\em unambiguous} (that is, any accepted data word has exactly one accepting run).
While it is known that when registers can hold terms, unambiguous machines are no more expressive than
deterministic ones~\cite{AdA2017STT}, we prove that when registers can hold only values that are updated in a copyless manner,
unambiguous machines are strictly more expressive than their deterministic counterparts.
The complexity of evaluation for unambiguous CRAs is only a constant (equal to the number of states of the machine)
factor more than the deterministic ones.
Finally, we consider both copyless and copyful (that is, unrestricted) updates since a CRA has a constant
number of registers that hold data values and get updated by executing a constant number of operations
while processing each data item. The number of bits needed to store data values can grow with the length
of the input stream when the data domain is unbounded, and the exact complexity of the streaming algorithm
for evaluation depends on the nature of operations and implementation details.

We call the class of transductions over data words defined by such unambiguous CRAs 
with registers containing only values and copyful updates as {\em streamable regular transductions} ($\SR$),
and the subclass when updates are required to be copyless as {\em streamable linear regular transductions} ($\SLR$).
Both these classes turn out to be closed under operations relevant for design of query languages.
For instance, consider the {\tt split} operation, the quantitative analog 
of the (unambiguous) language concatenation: given two partial transductions $f$ and $g$ over data words
and a binary data operation {\it op}, ${\tt split}(f,g,{\it op})$ is defined on an input string $w$ if it can be 
split uniquely into two parts $w=w_1 w_2$ such that both $f(w_1)$ and $g(w_2)$ are defined, and if so returns
${\it op}(f(w_1),g(w_2))$.
Both classes $\SR$ and $\SLR$ are closed under this operation.
It is worth noting that the class of general streaming algorithms with an appropriate complexity bound is not closed under this operation.
That is, it is possible that both $f$ and $g$ can be computed in a streaming fashion using memory that is logarithmic
in the size of the input string, but ${\tt split}(f,g,{\it op})$ requires memory that grows
linearly with the length of the input string (Theorem~\ref{thm:general-streaming-algorithms-not-closed}).
This justifies the role of regularity in the design of streaming algorithms in a modular fashion.

In section~3, we give a logical characterization of the streamable regular transductions.
Recall that the class of regular transductions corresponds to MSO-definable string-to-tree
transformations. We show that streamability corresponds to the requirement that in the output graph
edges cannot go backwards: the output graph has multiple copies corresponding to an index position $i$ 
in the input string, and the children of each such node must be copies corresponding to input  positions less
than or equal to $i$.
Furthermore, allowing updates in a CRA to be copyful corresponds to allowing sharing in the output graph
as long as no cycles are created. More precisely,
the class $\SR$ corresponds to MSO-definable transformations from string graphs to
directed acyclic graphs without backward edges, and
the class $\SLR$ corresponds to MSO-definable transformations from string graphs to
trees without backward edges.
The proof builds on ideas developed in \cite{BE2000, EM1999, AdA2017STT}, but is self-contained.

In section~4, we study the relationship with {\em weighted automata\/}, an extensively studied 
quantitative extension of automata~\cite{DKV2009HWA}.
A weighted automaton defines a (total) function from $\Sigma\kstar$ to a set $D$ equipped with two
operations, product and sum, that form a semiring. Such an automaton is a nondeterministic 
finite-state machine, each transition of which is labeled with a tag in $\Sigma$ and a weight in $D$.
The cost of a run of the automaton is the product of all the weights of the transitions in the run,
and the cost associated with a string is the sum of the costs of all the corresponding runs of the automaton.
If we choose the set of operations to contain only multiplication by constants, for each value in $D$,
then we show that the two classes $\SR$ and $\SLR$ coincide, and correspond exactly to unambiguous weighted automata.
If we choose the set of operations to contain multiplication by constants, for each value in $D$, and sum,
then we show that the class $\SR$ corresponds exactly to transductions definable by weighted automata.

Section~\ref{sec:related} contains an extensive discussion of related work, including weighted automata, register automata, and MSO-definable transformations.

\section{Cost Register Automata}
\label{sec:CRAs}

Cost Register Automata (CRAs) were introduced in \cite{AdADRY2013CRA} with the goal of providing a machine-based characterization of the class of \emph{regular transductions} (called \emph{regular cost functions} in \cite{AdADRY2013CRA}) from strings to costs. To capture precisely the class of transductions considered in \cite{AdADRY2013CRA}, it is sufficient to restrict the updates so that they are ``copyless'' (each variable is used at most once on the right-hand size of a parallel update), but it is necessary to allow variables to hold terms with at most one ``hole'' (parameter that can be used for substitution). For the classes of transductions that we consider here, we consider CRAs that can hold only values and the copylessness restriction is dropped to accommodate word-to-DAG transformations (defined later in section~\ref{sec:MSO}).

We will start by introducing the semantic objects of data transductions, which represent stream transformations, and several regular combinators on these objects. Then, we will present the formal computational model of (nondeterministic, copyful) Cost Register Automata (NCRAs) and illustrate its use with several examples. We define the class of \emph{streamable regular} (resp., \emph{streamable linear regular}) transductions, denoted $\SR$ (resp., $\SLR$), as the class of transductions that are computed by copyful (resp., copyless) unambiguous CRAs. The restriction of ``copylessness'' essentially says that it is not allowed to copy the contents of registers. A key feature of the CRA model is that it is parameterized by the operations on data values that are allowed when updating the registers.

We will also show that in the case of copyless CRAs, unambiguous machines are more expressive than their deterministic counterparts. It will then be established that both classes $\SR$ and $\SLR$ are closed under regular combinators that are straightforward analogs of the familiar regular combinators for word languages. We will conclude with a short discussion of the complexity of evaluation in the CRA model.

\subsection{Data Transductions and Combinators}
\label{subsec:transductions}

We fix a finite alphabet $\Sigma$ whose elements are called \emph{tags}, a data set $D$ of \emph{data values}. A \emph{data word} is a sequence of tagged values, i.e.\ a word over the alphabet $\Sigma \times D$. A \emph{data transduction} is a partial function of type
\[
  (\Sigma \times D)\kstar \pto D.
\]
For a data word $w \in (\Sigma \times D)\kstar$, we write $w|_\Sigma$ to denote the elementwise projection of $w$ to the tag component. More formally,
\[
  (a_1, d_1) \sspc
  (a_2, d_2) \sspc
  \cdots \sspc
  (a_n, d_n)
  \,|_\Sigma =
  a_1 a_2 \cdots a_n.
\]
We will only consider transductions $f$ that satisfy the following property: for all input sequences $u, v$ with $u|_\Sigma = v|_\Sigma$, $f(u)$ is defined iff $f(v)$ is defined. The \emph{rate} of $f$ is the language $\rate(f) \subseteq \Sigma\kstar$ defined as follows:
\[
  \rate(f) = \{ \sigma \in \Sigma\kstar \mid
    \text{for every $w$ with $w|_\Sigma = \sigma$, $f(w)$ is defined}
  \}.
\]
So, the domain of a transduction is equal to
\[
  \dom(f) = \{
    w \in (\Sigma \times D)\kstar \mid
    w|_\Sigma \in \rate(f)
  \}.
\]
A transduction $f: (\Sigma \times D)\kstar \pto D$ is said to be \emph{value-oblivious} if for all input sequences $u, v$ with $u|_\Sigma = v|_\Sigma$ it satisfies $f(u) = f(v)$. So, a value-oblivious transduction $f$ can be represented by a partial function $\hat f: \Sigma\kstar \pto D$ whose domain is equal to the rate of $f$ and which satisfies $f(w) = \hat f(w|_\Sigma)$ for every input sequence $w$.

Now, we also fix a family $\O$ of constants and operations that are allowed on the data set $D$. We write $\O_n$ to denote the set of $n$-ary operations that are contained in $\O$. In particular, $\O_0$ is the set of constants in $\O$. We will consider below several combinators for composing data transductions using the allowed operations. Before we give the definitions of these combinators, we need to introduce some additional notation.
The operation $\odot$ is \emph{unambiguous concatenation} of strings and ${}\kostar$ is \emph{unambiguous iteration}, which are defined as follows:
\begin{align*}
A \odot B &=
\{ w \mid
  \text{there are unique $u \in A$ and $v \in B$ with $w = uv$}
\}
\\
A\kostar &=
\{ w \mid
  \text{there are unique $n \geq 0$ and $w_1, \ldots, w_n \in A$}
\\ &\hspace{3.5em}
\text{with $w = w_1 \cdots w_n$} \}
\end{align*}
The \emph{left fold} combinator $\fold: (B \times A \to B) \times B \times A\kstar \to B$ is defined by
\begin{align*}
\fold(f, b, \eps) &= b
\ \text{and}\ 
\\
\fold(f, b, a \cdot w) &= \fold(f, f(b,a), w)
\end{align*}
for all $f: B \times A \to B$, $b \in B$, $a \in A$, and $w \in A\kstar$. For example:
\begin{align*}
\fold(f, b, a_1) &= f(b, a_1)
\\
\fold(f, b, a_1 a_2) &= f(f(b, a_1), a_2)
\\
\fold(f, b, a_1 a_2 a_3) &= f(f(f(b, a_1), a_2), a_3)
\end{align*}
Table~\ref{table:combinators} contains the typing rules and the definitions of several combinators for data transductions: output combination, choice ($\elseQ$), quantitative concatenation ($\splitQ$), and quantitative iteration ($\iterQ$). These combinators are quantitative analogs of the familiar combinators for string languages: intersection, union, concatenation, and iteration. These regular operations on data transductions are relevant for the design of query languages for stream processing, since they enable the modular description of complex streaming computations.

\begin{table}
\centering\small
$\begin{gathered}
\AxiomC{$\op: D^n \to D$}
\AxiomC{$f_i: (\Sigma \times D)\kstar \pto D$ for all $i$}
\RightLabel{(output combination)}
\BinaryInfC{$\begin{aligned}
\op(f_1,\ldots,f_n) &: (\Sigma \times D)\kstar \pto D
\\
\rate(\op(f_1,\ldots,f_n)) &= \rate(f_1) \cap \cdots \cap \rate(f_n)
\\
\op(f_1,\ldots,f_n)(w) &= \op(f_1(w),\ldots,f_n(w))
\end{aligned}$}
\DisplayProof
\\[4ex]
\AxiomC{$f, g: (\Sigma \times D)\kstar \pto D$}
\RightLabel{(choice)}
\UnaryInfC{$\begin{aligned}
f \elseQ g &: (\Sigma \times D)\kstar \pto D
\\
\rate(f \elseQ g) &= \rate(f) \cup \rate(g)
\\
(f \elseQ g)(w) &= \begin{cases}
  f(w), &\text{if $w|_\Sigma \in \rate(f)$} \\
  g(w), &\text{otherwise}
\end{cases}
\end{aligned}$}
\DisplayProof
\\[4ex]
\AxiomC{$f, g: (\Sigma \times D)\kstar \pto D$}
\AxiomC{$\op: D \times D \to D$}
\RightLabel{$\begin{aligned}
  (&\text{quantitative} \\[-0.5ex]
  &\text{concatenation})
\end{aligned}$}
\BinaryInfC{$\begin{aligned}
\splitQ(f, g, \op)&: (\Sigma \times D)\kstar \pto D
\\
\rate(\splitQ(f, g, \op)) &= \rate(f) \odot \rate(g)
\\
\splitQ(f, g, \op)(u v) &= \op(f(u), g(v)),\ \text{if $u|_\Sigma \in \rate(f)$ and $v|_\Sigma \in \rate(g)$}
\hspace{-3em}
\end{aligned}$}
\DisplayProof
\\[4ex]
\AxiomC{$f: (\Sigma \times D)\kstar \pto D$}
\AxiomC{$c \in D$}
\AxiomC{$\op: D \times D \to D$}
\RightLabel{$\begin{aligned}
  (&\text{quantitative} \\[-0.5ex]
  &\text{iteration})
\end{aligned}$}
\TernaryInfC{$\begin{aligned}
\iterQ(f, c, \op) &: (\Sigma \times D)\kstar \pto D
\\
\rate(\iterQ(f, c, \op)) &= \rate(f)\kostar
\\
\iterQ(f, c, \op)(w_1 \cdots w_n) &=
\fold(\op, c, [f(w_1),\ldots,f(w_n)]),
\\ &\hspace{1.4em}
\text{if $w_i|_\Sigma \in \rate(f)$ for all $i=1,\ldots,n$}
\end{aligned}$}
\DisplayProof
\end{gathered}$
\caption{Some regular combinators for data transductions: output combination, choice, quantitative concatenation, and quantitative iteration.
Here, $\odot$ is \emph{unambiguous concatenation} and ${}\kostar$ is \emph{unambiguous iteration}, so that quantitative concatenation (resp., quantitative iteration) is defined only on input strings where the concatenation (resp., iteration) of the underlying rates $\rate(-)$ is unambiguous.
}
\label{table:combinators}
\end{table}

The work \cite{AFR2016QRE} introduces a family of regular combinators that captures a class of transductions that is defined in terms of CRAs whose registers use terms with variables instead of values. The StreamQRE language of \cite{MRAIK2017SQRE, AM2017SQRE} extends the combinators of Table~\ref{table:combinators} with additional constructs (such as streaming composition and key-based partitioning) that are useful in practical stream processing applications.

\subsection{Syntax and Semantics of CRAs}

A Cost Register Automaton (CRA) is a machine that maps data words, i.e.\ strings over an input alphabet of tagged values, to output values. It uses a finite-state control and a finite set of registers that contain values. At each step, the machine reads an input tag and an input value, updates its control state, and updates its registers using a parallel assignment. The definition of the machine is parameterized by the set of expressions that can be used in the assignments. For example, a CRA with increments can use multiple registers to compute alternative costs, perform updates of the form
$x := y + c$, where $x$ and $y$ are registers and $c$ is a constant, at each step, and commit to the cost computed in one of the registers at the end.

Let $X$ be a set of variables, and suppose that $\O$ is the family of available constants and operations on the data set $D$. The set of \emph{expressions} over $X$, denoted $\expr_\O[X]$ is defined by the following rules:
\begin{gather*}
\AxiomC{constant $c \in \O$}
\UnaryInfC{$c: \expr_\O[X]$}
\DisplayProof
\qquad
\AxiomC{variable $x \in X$}
\UnaryInfC{$x: \expr_\O[X]$}
\DisplayProof
\\[1ex]
\AxiomC{$n$-ary operation $\op \in \O$}
\AxiomC{$t_i: \expr_\O[X]$ for all $i=1,\ldots,n$}
\BinaryInfC{$\op(t_1,\ldots,t_n): \expr_\O[X]$}
\DisplayProof
\end{gather*}
A function $\alpha: X \to D$ is called a \emph{variable assignment}, and it extends uniquely to a homomorphism (i.e., operation-preserving mapping) $\hat\alpha: \expr_\O[X] \to D$. An expression $t \in \expr_\O[X]$ denotes a function $\sem{t}: D^X \to D$, defined as follows: for a variable assignment $\alpha: X \to D$, we put $\sem{t}(\alpha) = \hat\alpha(t)$. That is, $\sem{t}(\alpha)$ is the result of evaluating the variable-free expression that is obtained from $t$ by replacing each variable $x$ by $\alpha(x)$.

In the development that follows, we will also consider a special symbol $\val$, which should be assumed to be always distinct from all variables in $X$. The symbol $\val$ is meant to be used as a reference to the value of the current data item. An expression $t \in \expr_\O[X \cup \{\val\}]$ that potentially contains occurrences of $\val$ denotes a function $\sem{t}: D^X \times D \to D$, given by
\[
  \sem{t}(\alpha,d) = \hat\beta(t)
  \ \text{where $\beta = \alpha[\val \mapsto d]: X \cup \{\val\} \to D$}
\]
for every variable assignment $\alpha: X \to D$ and every value $d \in D$. We use the notation $\alpha[\val \mapsto d]$ to mean the extension of $\alpha$ that maps $\val$ to $d$.

\begin{definition}[CRA]
\label{def:CRA}
\normalfont
A (nondeterministic, copyful) \emph{Cost Register Automaton} (\emph{NCRA}) over the tag alphabet $\Sigma$, data values $D$, and data operations $\O$ is a tuple
\[
  \aut A =
  (Q,X,\Delta,I,F),
\]
where
\begin{itemize}
\item
$Q$ is a finite set of states,
\item
$X$ is a finite set of registers,
\item
$\Delta \subseteq Q \times \Sigma \times U_\O \times Q$ is the set of transitions with $U_\O$ the set of register updates $X \to \expr_\O[X \cup \{\val\}]$,
\item
$I: Q \pto (X \to \expr_\O[\emptyset])$ is the initialization function, and
\item
$F: Q \pto \expr_\O[X]$ is the finalization function.
\end{itemize}
The domain of $I$ is the set of \emph{initial states}, and the domain of $F$ is the set of \emph{final} or \emph{accepting states}. A CRA is said to be \emph{unambiguous} (resp., deterministic) if the underlying NFA is unambiguous (resp., deterministic). We use the abbreviations NCRA, UCRA, and DCRA to indicate that an automaton is nondeterministic, unambiguous, and deterministic respectively.

A CRA is said to be \emph{copyless} if it satisfies the following restrictions:
\begin{inparaenum}[(1)]
\item
for every transition $(p,a,\theta,q) \in \Delta$ and every register $x$, there is at most one occurrence of $x$ in the list of expressions $\theta(x_1), \ldots, \theta(x_n)$, where $x_1, \ldots, x_n$ is an enumeration of $X$, and
\item
for every final state $q$ and every register $x$, there is at most one occurrence of $x$ in the expression $F(q)$. Note that we do \emph{not} require that the special symbol $\val$ is only used once in each expression. This is reasonable because $\val$ represents the current value in the input data word, which changes at each computation step of the automaton. Therefore, a particular value in the input data word will only be used (copied) at most a constant number of times.
\end{inparaenum}

A CRA is said to be \emph{trim} if the satisfies the properties: (1) every state of the automaton is reachable from an initial state, and (2) some final state is reachable from every state.
\end{definition}

\paragraph{\bf\em Semantics of CRAs}
A CRA computes like a classical finite-state automaton, but the configuration consists of both a current state and a register assignment $X \to D$. A transition specifies the next state, as well as the update of the registers using their current values and the current data value from the input. For an input sequence $w = (a_1,d_1) \sspc (a_2,d_2) \sspc \cdots \sspc (a_n,d_n)$ in $(\Sigma \times D)\kstar$, we define a \emph{$w$-run} in $\aut A$ to be a sequence
\[
  (q_0,\alpha_0)
  \xlongrightarrow{(a_1,d_1)} (q_1,\alpha_1)
  \xlongrightarrow{(a_2,d_2)} (q_2,\alpha_2)
  \xlongrightarrow{(a_3,d_3)} \cdots
  \xlongrightarrow{(a_n,d_n)} (q_n,\alpha_n)
\]
with $q_i \in Q$ and $\alpha_i: X \to D$ for all $i$ so that the following hold:
\begin{enumerate}
\item
\emph{Initialization}:
$q_0$ is an initial state and $\alpha_0(x) = \sem{I(q_0)(x)}$ for every register $x \in X$.
\item
\emph{Transition}: for every transition $(p,\alpha) \to^{(a,d)} (q,\beta)$ of the run there is a register update function $\theta \in U_\O$ with $(p,a,\theta,q) \in \Delta$ s.t.
\[
  \beta(x) = \sem{\theta(x)}(\alpha,d)
  \ \text{for every register $x \in X$}.
\]
\item
\emph{Finalization}: $q_n$ is a final state.
\end{enumerate}
The \emph{value} of the run is $\sem{F(q_n)}(\alpha_n)$. A nondeterministic CRA implements a multi-valued transduction $ (\Sigma \times D)\kstar \to \mathcal{P}(D)$, where $\mathcal{P}$ is the powerset operator. The value of this transduction on an input sequence $w$ is the set of values of all $w$-runs. An unambiguous CRA has at most one run for each input sequence, and therefore implements a transduction $(\Sigma \times D)\kstar \pto D$. We write $\sem{\aut A}$ to denote the transduction that an automaton $\aut A$ implements.

As for classical NFAs, we can think that the computation of a CRA proceeds by placing \emph{tokens} on the states. A token on a state $q$ indicates that there is a run from an initial state to $q$ after consuming the given input prefix. When a state has a token on it we say that it is \emph{active}. The computation starts by placing a token on each of the initial states. Every time an item is consumed, each token is transformed into a set of new tokens according to the transition relation. In the case of plain NFAs, a token carries no information. For CRAs, on the other hand, each token carries a register assignment $X \to D$. For this reason, it is possible to have more than one token per state in the case of nondeterministic (ambiguous) CRAs. For unambiguous trim CRAs, at any time during the computation there can be at most one token per state.

\begin{definition}[Streamable Regular Transductions]
\label{def:streamable}
\normalfont
Fix a tag alphabet $\Sigma$, a set $D$ of values, and a family $\O$ of operations over $D$. We define the class of \emph{streamable regular transductions} (over $\O$), denoted $\SR(\O)$, to be the class of data transductions that can be computed by some UCRA (over $\O$). Similarly, we define the class of \emph{streamable linear regular transductions}, denoted $\SLR(\O)$, to be the class of data transductions that can be computed by some copyless UCRA. Table~\ref{table:classesCRA} summarizes these definitions.
\end{definition}

\begin{table}[]
\centering
\begin{tabular}{c|c}
Unambiguous CRA & Class of transductions
\\ \hline
copyless
& $\SLR$: Streamable Linear Regular
\\
copyful
& $\SR$: Streamable Regular
\end{tabular}
\caption{Classes of transductions computed by unambiguous CRAs.}
\label{table:classesCRA}
\end{table}

\subsection{Examples}
\label{subsec:examples}

We present in this subsection several examples that illustrate how CRAs compute. We consider both deterministic and unambiguous variants, and we give examples of copyless and copyful (i.e., unrestricted) register updates. Whenever convenient, we will allow the use of $\eps$-transitions for the unambiguous variants. This is w.l.o.g.\ because $\eps$-transitions can be eliminated with a straightforward variant of the $\eps$-elimination procedure for classical finite-state automata (assuming there are no $\eps$-cycles). 

\begin{example}[Copyless DCRA]
\label{ex:copylessDCRA}
\normalfont
Suppose the tag alphabet is $\Sigma = \{ a,b \}$, the type of values is $D = \N$ (the set of nonnegative integers), and $\O$ consists of the constant $0$ and the binary addition operation. Consider the following transduction $f: (\Sigma \times D)\kstar \pto D$, which is defined on all nonempty sequences. If a sequence ends with an $a$-labeled value, then $f$ outputs the sum of all $a$-labeled values in the sequence. Similarly, if a sequence ends with a $b$-labeled value, then $f$ outputs the sum of all $b$-labeled values in the sequence. Then, $f$ is implemented by the following copyless CRA:
\begin{gather*}
\begin{tikzpicture}[node distance=3cm, ->]
\small
\node (i) {};
\node [state, below of=i, node distance=1.25cm] (p) {$p$};
\node [below of=p, node distance=2cm] (p2) {};
\node [state, accepting, rounded rectangle, left of=p2] (qa) {$q_a \mid x$};
\node [state, accepting, rounded rectangle, right of=p2] (qb) {$q_b \mid y$};
\path (i) edge node[right] {$\theta$} (p);
\path (p) edge node[above, yshift=1ex] {$a \mid \theta_a$} (qa);
\path (p) edge node[above, yshift=1ex] {$b \mid \theta_b$} (qb);
\path (qa) edge[loop left] node[left] {$a \mid \theta_a$} (qa);
\path (qb) edge[loop right] node[right] {$b \mid \theta_b$} (qb);
\path (qa) edge[bend left=10] node[above] {$b \mid \theta_b$} (qb);
\path (qb) edge[bend left=10] node[below] {$a \mid \theta_a$} (qa);
\end{tikzpicture}
\end{gather*}
where the initialization $\theta$ and the register updates $\theta_a$, $\theta_b$ are defined as follows:
\begin{align*}
\theta &= \begin{cases}
  x := 0 \\
  y := 0
\end{cases}
&
\theta_a &= \begin{cases}
  x := x + \val \\
  y := y
\end{cases}
&
\theta_b &= \begin{cases}
  x := x \\
  y := y + \val
\end{cases}
\end{align*}
In the diagram, we have indicated the initialization function with the arrow labeled with $\theta$, and the finalization function is given by the annotations $x$ and $y$ to the two accepting states (shown with a doubly circled border).

The register $x$ holds the sum of all $a$-labeled values seen so far, and the register $y$ holds the sum of all $b$-labeled values seen so far. The state $p$ is active only upon initialization, and $q_a$ (resp., $q_b$) is active when the input sequence ends with an $a$-labeled (resp., $b$-labeled) value. The CRA is deterministic. Moreover, it is copyless, since every register is used at most once in the right-hand side of the updates $\theta_a$ and $\theta_b$.
\end{example}

\begin{example}[Copyless UCRA]
\label{ex:copylessUCRA}
\normalfont
In Example~\ref{ex:copylessDCRA}, it was shown how to implement the transduction $f$ using a copyless DCRA with two registers. We will show now how to implement the same transduction using a copyless UCRA with one register. Informally, the idea is to employ unambiguous nondeterminism in order to guess whether the input data word will end with the tag $a$ or the tag $b$. We assume additionally that $\eps$-transitions are allowed.
\begin{gather*}
\begin{tikzpicture}[node distance=3.5cm, ->]
\footnotesize
\node (i) {};
\node [state, below of=i, node distance=1.25cm] (p) {$p$};
\node [below of=p, node distance=1.25cm] (p2) {};
\node [state, left of=p2, node distance=2.5cm] (pa) {$p_a $};
\node [state, accepting, rounded rectangle, left of=pa] (qa) {$q_a \mid x$};
\node [state, right of=p2, node distance=2.5cm] (pb) {$p_b$};
\node [state, accepting, rounded rectangle, right of=pb] (qb) {$q_b \mid x$};
\path (i) edge node[right] {$x := 0$} (p);
\path (p) edge node[below, pos=0.3, yshift=-1.5ex] {$\eps \mid x := 0$} (pa);
\path (p) edge node[below, pos=0.3, yshift=-1.5ex] {$\eps \mid x := 0$} (pb);
\path (pa) edge[loop above] node[above] {$a \mid x := x + \val$} (pa);
\path (pa) edge[loop below] node[below] {$b \mid x := x$} (pa);
\path (pa) edge node[above] {$a \mid x := x + \val$} (qa);
\path (pb) edge[loop above] node[above] {$b \mid x := x + \val$} (pb);
\path (pb) edge[loop below] node[below] {$a \mid x := x$} (pb);
\path (pb) edge node[above] {$b \mid x := x + \val$} (qb);
\end{tikzpicture}
\end{gather*}
The register $x$ holds the sum of all $a$-labeled values at the states $p_a$ and $q_a$, and the sum of all $b$-labeled values at the states $p_b$ and $q_b$. The automaton is copyless, since the register $x$ is used at most once in the right-hand side of the updates.
\end{example}

\begin{example}[Copyless DCRA]
\label{ex:twoLevels}
\normalfont
Suppose the tag alphabet is $\Sigma = \{ a,\# \}$, the type of values is $D = \N$, and $\O$ consists of the constant $0$, the binary addition operation, and $\max$. Consider the following transduction $f: (\Sigma \times D)\kstar \pto D$, whose rate is given by the regex $(a\kplus \cdot \#)\kstar$. A maximal subsequence of the input that is of the form $a a \ldots a \#$ is called a \emph{block}. The transduction $f$ outputs the maximum cost over all input blocks, where the cost of a block is the sum of the $a$-labeled values. This is implemented by the following two-register copyless CRA:
\begin{gather*}
\begin{tikzpicture}[node distance=4cm, ->]
\small
\node (i) {};
\node [state, accepting, rounded rectangle, right of=i, node distance=1.25cm] (p) {$p \mid y$};
\node [state, right of=p] (q) {$q$};
\path (i) edge node[above] {$\theta$} (p);
\path (p) edge[bend left=12] node[above] {$a \mid \theta_a$} (q);
\path (q) edge[bend left=12] node[below] {$\# \mid \theta_\#$} (p);
\path (q) edge[loop right] node[right, xshift=1ex] {$a \mid \theta_a$} (q);
\end{tikzpicture}
\end{gather*}
where the initialization $\theta$ and the register updates $\theta_a$, $\theta_\#$ are given by:
\begin{align*}
\theta &= \begin{cases}
  x := 0 \\
  y := 0
\end{cases}
&
\theta_a &= \begin{cases}
  x := x + \val \\
  y := y
\end{cases}
&
\theta_\# &= \begin{cases}
  x := 0 \\
  y := \max(y,x)
\end{cases}
\end{align*}
The register $x$ holds the sum of all $a$-labeled values in the current block, and the register $y$ holds the maximum cost over all complete blocks seen so far. The CRA is deterministic and copyless (every register is used at most once in the right-hand side of the updates $\theta_a$ and $\theta_\#$).
\end{example}

\begin{example}[Copyful UCRA]
\label{ex:copyfulUCRA}
\normalfont
Suppose the tag alphabet is $\Sigma = \{ a,b \}$, the type of values is $D = \N$, and $\O$ consists of the constant $0$ and the binary operations $\max$ and $\ominus$, where the latter is natural number subtraction, sometimes called ``monus'', and defined as $x \ominus y = \max(x - y, 0)$ for all $x, y \in \N$. Consider the following transduction $f: (\Sigma \times D)\kstar \pto D$, which is defined on sequences that contain at least one $a$-labeled value. For these sequences, the value of $f$ is the \emph{maximum drawdown} (the maximum loss from a peak to a trough) in the input signal after the last occurrence of a $b$-labeled value. The transduction $f$ is implemented by the following copyful CRA:
\begin{gather*}
\begin{tikzpicture}[node distance=3.5cm, ->]
\small
\node (i) {};
\node [state, right of=i, node distance=1.25cm] (p) {$p$};
\node [state, accepting, rounded rectangle, right of=p] (q) {$q \mid y$};
\path (i) edge node[above] {$\theta_0$} (p);
\path (p) edge node[above] {$b \mid \theta_0$} (q);
\path (p) edge[loop above] node[left, xshift=-1ex] {$a,b \mid \theta_0$} (p);
\path (q) edge[loop above] node[right, xshift=1ex] {$a \mid \theta$} (q);
\end{tikzpicture}
\end{gather*}
where the initialization $\theta_0$ and the register update $\theta$ are defined as follows:
\begin{align*}
\theta_0 &= \begin{cases}
  x := 0 \\
  y := 0
\end{cases}
&
\theta &= \begin{cases}
  x := \max(x, \val) \\
  y := \max(y, \max(x, \val) \ominus \val)
\end{cases}
\end{align*}
The numerical computation on the values concerns the part of the input after the last occurrence of a $b$-labeled value. The register $x$ holds the maximum value, and the register $y$ holds the maximum drawdown. The state $p$ is always active, and $q$ is active when at least one $b$-labeled value is seen in the input. The CRA is unambiguous. Moreover, it is not copyless, since the register $x$ is used twice in the right-hand sides of the register update $\theta$.
\end{example}

\begin{example}[CRAs with holes]
\normalfont
Suppose the tag alphabet is $\Sigma = \{ a \}$, the type of values is $D = \Q$ (the set of rational numbers), and $\O$ consists of the constant $0$ and the binary operation $\op$ given by $\op(x,y) = \lambda \cdot x + y$, where $\lambda$ is a fixed rational constant in the open interval $(0,1)$. The transduction $f$ maps a data word $w$ with $w|_D = d_1 d_2 \ldots d_n \in D\kstar$ to $\lambda^{n-1} \cdot d_1 + \cdots + \lambda \cdot d_{n-1} + d_n$. We write $w|_D$ to denote the elementwise projection of the data word to the value component. The transduction can be computed by the following CRA:
\begin{gather*}
\begin{tikzpicture}[node distance=4cm, ->]
\small
\node (i) {};
\node [state, accepting, rounded rectangle, right of=i, node distance=2.5cm] (q) {$p \mid x$};
\path (i) edge node[above] {$x := 0$} (q);
\path (q) edge[loop right] node[right, xshift=1ex] {$a \mid x := \op(x, \val) = \lambda \cdot x + \val$} (q);
\end{tikzpicture}
\end{gather*}
Define now the transduction $g$ as $g(w_1 w_2 \ldots w_n) = f(w_n \ldots w_2 w_1)$. The transduction $g$ cannot be computed by a CRA (over $\O$), because the registers can only hold values. However, if the registers were allowed to hold terms with a hole $\Box$ (a placeholder that can be substituted later in the computation), then $g$ could be expressed.
\begin{gather*}
\begin{tikzpicture}[node distance=4cm, ->]
\small
\node (i) {};
\node [state, accepting, rounded rectangle, right of=i, node distance=3cm] (q) {$p \mid x[0/\Box]$};
\path (i) edge node[above] {$x := \Box$} (q);
\path (q) edge[loop right] node[right, xshift=1ex] {$a \mid x := x[\op(\Box,\val)/\Box]$} (q);
\end{tikzpicture}
\end{gather*}
In the machine shown above the register is meant to hold a term (with potential occurrences of $\Box$), and the operation $x[t/\Box]$ denotes the result of substituting $t$ for $\Box$ in the term that $x$ holds. For example, if the input is $w = (a,d_1) \sspc (a,d_2) \sspc (a,d_3)$, then the successive contents of the register $x$ (after 3 steps of computation) are shown below:
\begin{align*}
\Box &&
\op(\Box,d_1) &&
\op(\op(\Box,d_2),d_1) &&
\op(\op(\op(\Box,d_3),d_2),d_1)
\end{align*}
So, the ability to use the holes $\Box$ and substitute terms into it provides additional computation power. The disadvantage of using terms is that in some cases their size can grow linearly in the size of the input.
\end{example}

\subsection{Choice of Operations on Data Values}
\label{sec:choiceOps}

The definition of CRAs is parameterized by the choice of operations that are allowed when updating the registers. We now emphasize that seemingly inconsequential variations of the available operations can change what can be computed by a CRA. Suppose that $M = (D,\cdot,1)$ is a monoid, i.e.\ $\cdot$ is an associative binary operation on $D$ and $1$ is a left and right identity for the $\cdot$ operation. Define $\O_M$ to be the family of operations that contains $1$ and the binary operation $\cdot$, and $\O^\unary_M$ to contain the identity $1$ and the unary operation $(- \cdot d)$ for every $d \in D$. (Note that we include only unary multiplication on the right, not on the left---this is not crucial for the present discussion, but will be necessary for a later result about $\O^\unary_M$, Theorem~\ref{thm:UWA}.) The following lemma applies to $\O^\unary_M$:

\begin{lemma}
\label{lemma:monoidWeak}
\normalfont
Suppose that $\O$ consists only of unary operations and constants.
Then the class of transductions $\SR(\O)$ is equal to $\SLR(\O)$.
\end{lemma}
\begin{proof}
It suffices to see that a copyful UCRA over $\O$ can be simulated by a copyless UCRA over $\O$. The key observation is that in a copyful UCRA with only unary operations, at most one register can contribute to the final output, because there is no way to combine multiple registers. So at every step of the computation, the machine can guess the register $x$ that will contribute to the final output (using unambiguous nondeterminism), and proceed by updating only $x$. In particular, if the copyful UCRA over $\O$ has states $Q$ and registers $X$, we construct a copyless UCRA over $\O$ with states $Q \times (X \cup \{\bot\})$ and only one register, where the single register holds the value of the register $x$ that we guess will contribute to the final output, and the second component of the state tracks which register it is in the original machine---or $\bot$ if we guess that none of the current registers will contribute to the output. Given any UCRA over $\O$, this construction produces an equivalent copyless UCRA with only one register.
\end{proof}

The above lemma shows that if $M = (D,\cdot,1)$ is a monoid, then the class of transductions $\SR(\O^\unary_M)$ is equal to $\SLR(\O^\unary_M)$. However, in general, the class of transductions $\SR(\O)$ may be strictly larger than $\SLR(\O)$. In particular, consider the family of operations $\O_M$, where $M = (D,\cdot,1)$ is the free monoid generated by the alphabet $\Sigma$. That is, $D$ is the set of finite words over $\Sigma$, $\cdot$ is word concatenation, and $1$ is the empty word. For all transductions on $\SLR(\O_M)$, the output is of size linear in the size of the input. This property does not hold, however, for the transductions of $\SR(\O_M)$. The copyful single-state single-register CRA
\begin{gather*}
\begin{tikzpicture}[->]
\small
\node (i) {};
\node [state, accepting, rounded rectangle, right of=i, node distance=2.5cm] (q) {$q \mid x$};
\path (i) edge node[above] {$x := a$} (q);
\path (q) edge[loop right] node[right] {$a,b \mid x := x \cdot x$} (q);
\end{tikzpicture}
\end{gather*}
emits outputs that are of size exponential in the input. Lemma~\ref{lemma:monoidWeak} and these observations establish that the precise choice of $\O$ can affect the class of transductions that can be computed.

Another important point about the set $\O$ of operations is that it can incorporate tests on values. The definition of CRAs does not allow tests on registers as guards on the transitions, which means that the finite control of the automaton depends only on the observed sequence of input tags. It is possible, however, to compute a transduction such as
\[
  f(w) = \begin{cases}
    \text{sum of $a$-tagged values}, &\text{if $|w|_a = |w|_b$} \\
    \text{sum of $b$-tagged values}, &\text{if $|w|_a \neq |w|_b$}
  \end{cases}
\]
(where $|w|_a$ is the number of occurrences of the tag $a$ in the data word $w$) by including in $\O$ an \emph{if-then-else} operation such as:
\[
  \ITE(x, y, z, w) = \begin{cases}
    z, &\text{if $x = y$} \\
    w, &\text{if $x \neq y$}
  \end{cases}
\]

\subsection{Unambiguity Versus Determinism}

In this section we explore the relationship between unambiguous and deterministic CRAs. The first result is a kind of subset construction for transforming a copyful UCRA to an equivalent copyful DCRA. This means that unambiguous nondeterminism is not essential when the automata are allowed to copy values. In the case of copyless CRAs, however, we will see that unambiguity is essential. That is, there are data transductions that can be computed by copyless UCRAs but not by copyless DCRAs.

To get an intuitive understanding of why copyless DCRAs are weaker, suppose that a copyless DCRA has a register $x$ which can contribute to the output in two different ways, depending on some regular property of the future input. Because the automaton is not allowed to copy the value of $x$, it cannot commit to any of the two different eventualities in order to update $x$ appropriately. If this situation can only arise a fixed number of times in the computation, a solution can be given by blowing up the number of registers: create from the outset two copies of the register $x$, one for each possible eventuality and update them separately. We will see, however, that there are examples where this situation can arise an arbitrary number of times. Such examples witness the power of unambiguity in the model of copyless CRAs.

\begin{theorem}[\bf Copyful UCRA to Copyful DCRA]
\label{thm:UCRAtoDCRA}
\normalfont
For any $\O$, every copyful unambiguous CRA over $\O$ is equivalent to a copyful deterministic CRA over $\O$.
\end{theorem}
\begin{proof}
The result holds trivially when $\O$ does not contain any constants. In this case, the initialization function then has to be undefined and hence the automaton can only implement the transduction that is undefined for every input, which can be implemented by a DCRA. So, we assume w.l.o.g. that $\O$ contains at least one constant. We will use this assumption in the rest of the proof in order to set registers to arbitrary constant values.

Let $\aut A = (Q,X,\Delta,I,F)$ be an arbitrary copyful unambiguous CRA. W.l.o.g.\ we assume that $\aut A$ is trim. We define the deterministic CRA $\aut B$ using a modified subset construction. The automaton $\aut B$ has registers $Q \times X$ and its state space is equal to
\[
  \{ \Delta(\dom(I),w) \mid w \in \Sigma\kstar \},
\]
where $\Delta(P,w)$ for $P \subseteq Q$ is the set of states reachable from a state in $P$ via a path that is labeled with the word $w$.

The initial state of $\aut B$ is equal to the set of initial states of $\aut A$, and a register $(q,x)$ of $\aut B$ is initialized to $I(q)(x)$ if $q$ is initial. The initialization of the rest of the registers can be chosen arbitrarily.

A state $P$ of $\aut B$ is final if $P$ contains a final state $q$ of $\aut A$. Since $\aut A$ is unambiguous, each state of $\aut B$ contains at most one final state of $\aut A$. The finalization term for $P$ is defined to be the term that results from $F(q)$ by substituting $(q,x)$ for every register $x$.

Consider now a state $P$ of $\aut B$ and a tag $a \in \Sigma$. Then, $R = \Delta(P,a) = \tbigcup_{p \in P} \Delta(p,a)$ is also a state of $\aut B$. We include a unique transition $P \to^a R$ in $\aut B$, which ensures that $\aut B$ is deterministic. It remains to specify the register update function $\eta$ for this transition. Consider an arbitrary $q \in R$. Since $\aut A$ is unambiguous and trim, there is a unique $p \in P$ with $q \in \Delta(p,a)$. For a register $(q, x)$ of $\aut B$ we put
\[
  \eta(q, x) =
  \text{the term resulting from $\theta(x)$ by replacing every register $y$ by $(p,y)$},
\]
where $\theta$ is the register update function of the unique transition $p \to^{a,\theta} q$ in the automaton $\aut A$. The update for registers not covered by the above description can be set arbitrarily.

In order to show that $\aut B$ is equivalent to $\aut A$, it suffices to establish that on any input, the sole active state of $\aut B$ encodes a copy of the registers of $\aut A$ for every active state of $\aut A$.
\end{proof}

Theorem~\ref{thm:UCRAtoDCRA} establishes the expressive equivalence of copyful UCRAs and copyful DCRAs. Now, we will show (Theorem~\ref{thm:separation}) that copyless UCRAs are strictly more expressive than copyless DCRAs (for a fixed family of data operations). To establish this separation result, it suffices to focus on CRAs that do not make use of the input data values in their computations. We say that a CRA is \emph{value-oblivious} if its register update function contains no occurrence of the $\val$ symbol. In this case the denotation of the CRA is a value-oblivious transduction, since only the input tags are used and the input data values do not affect the computation. For the rest of this subsection, we will be dealing exclusively with value-oblivious CRAs. For this reason, a data transduction is taken here to be a function of type $\Sigma\kstar \pto D$, and a CRA is implicitly assumed to not contain any occurrence of the $\val$ symbol. For the special case where $D = \Gamma\kstar$ for some finite alphabet $\Gamma$, a transduction $\Sigma\kstar \pto D$ is also called a \emph{string transduction} (the term \emph{word transduction} is often used in the relevant literature).

\newcommand{\f}{\mathsf{f}}
As an example of this, consider the following string transduction $\f: \Sigma\kstar \pto \Sigma\kstar$ with domain $\dom(\f) = \{ a,b \}\kplus \#$, where the alphabet is $\Sigma = \{ a,b,\# \}$:
\[
  \f(u \#) = \begin{cases}
    a^i \#, &\text{if $|u| = i$ and $u$ ends with $a$}
    \\
    b^i \#, &\text{if $|u| = i$ and $u$ ends with $b$}
  \end{cases}
\]
for all $u \in \{a,b\}\kplus$. We assume that the only data operation that is allowed is to concatenate string constants to the registers. An unambiguous CRA can compute this transduction with a single register by guessing the letter that will precede the $\#$ symbol. A deterministic CRA, on the other hand, requires at least two registers: one register for the case where the input ends in $a \#$, and one register for the case where the input ends in $b \#$.

This situation extends to the case where the computation of $\f$ is iterated several times. That is, let us define the string transduction $\f^k: \Sigma\kstar \pto \Sigma\kstar$ that iterates $\f$ sequentially $k$ times. The domain of $\f^k$ is equal to $\dom(\f^k) = (\{a,b\}\kplus \#)^k$, and $\f^k$ is given by
\[
  \f^k(u_1 \# \cdots u_k \#) = \f(u_1 \#) \cdots \f(u_k \#)
\]
for all $u_1, \ldots, u_k \in \{a,b\}\kplus$. In order to compute $\f^k$, an unambiguous CRA requires just one register: for every block of the input (where a block is a sequence of letters that ends with a $\#$ symbol), the automaton guesses the last letter of the block and extends the unique register with the appropriate letter. As we will show formally, however, a copyless DCRA requires at least $2^k$ registers: a separate register is needed for each possible combination of ending letters for the $k$ input blocks. For example, if $k = 2$ then there are $2^k = 4$ different cases for the form of the input strings in the domain of $\f^k$:
\begin{align*}
\ldots a \# \ldots a \#
&&
\ldots a \# \ldots b \#
&&
\ldots b \# \ldots a \#
&&
\ldots b \# \ldots b \#
\end{align*}
A copyless DCRA can only compute $\f^k$ by using one separate register for each one of these $2^k$ possibilities.

In the previous paragraph, we discussed how the iteration of $\f$ a fixed number $k$ of times creates the need for $2^k$ registers in the model of copyless DCRAs. An immediate consequence of that is that in this model, the iteration of $\f$ a non-constant number of times cannot be computed with a finite number of registers. This will establish the result that UCRAs are more expressive than DCRAs in the copyless case. Formally, the result is rather involved and requires a sequence of constructions on CRAs:
\begin{enumerate}
\item
Every CRA over \emph{unary} data operations is equivalent to the disjoint union of single-register CRAs.
\item
A copyless DCRA can be modified in a way that suppresses a prefix of the output.
\item
If a string transduction $g$ requires at least $k$ registers to be computed by a copyless DCRA, then the string transduction that computes $\f$ (defined previously) and then $g$ in sequence requires at least $2k$ registers in this model.
\end{enumerate}
Lemma~\ref{lemma:numberOfRegisters} is the heart of the argument. It establishes that the computation of the sequential iteration $\f^k$ (defined previously) requires at least $2^k$ registers in the model of copyless DCRAs.

We start the technical development with Lemma~\ref{lemma:simplification} below, which describes a construction for simplifying CRAs that use only unary operation on the values. For this special case, it is shown that the register updates are w.l.o.g.\ of the form $x := \delta(c)$ or $x := \delta(x)$, where $x$ is a register, $c$ is a constant, and $\delta$ is a unary operation. In other words, the CRA can be decomposed into the disjoint union of several single-register CRAs.

\begin{lemma}[Separable Updates]
\label{lemma:simplification}
\normalfont
Suppose that the collection of operations $\O$ consists only of unary operations and constants. Then, every copyless NCRA (resp., UCRA/DCRA) over $\O$ is equivalent to some copyless NCRA (resp., UCRA/DCRA) over $\O$ (with the same set of registers) whose register updates are of one of the following forms: $x := \delta(c)$ or $x := \delta(x)$, where $c$ is a constant and $\delta$ is the composition of unary operations.
\end{lemma}
\begin{proof}
We use the symbol $\delta$ in order to indicate an arbitrary composition of unary operation symbols. In a CRA $\aut A = (Q,X,\Delta,I,F)$ over $\O$ the register updates can be of any of the forms: $x := \delta(c)$, or $x := \delta(x)$, or $x := \delta(y)$ for $x \neq y$. We will employ a product construction in order to eliminate updates of the form $x := \delta(y)$ with $x \neq y$. If the assignment $x := \delta(y)$ is part of a parallel update, then $y$ cannot flow into another register due to the copyless restriction.

We will describe now the construction of the CRA $\aut B$, which has the same registers as $\aut A$. The idea is that $\aut B$ computes like $\aut A$, but also maintains some extra information in the finite control that allows to ``rename'' the registers. A renaming function $\sigma$ is a bijection on the set $X$ of registers, and the state space of $\aut B$ is the cartesian product of the state space of $\aut A$ with the function space of bijections $X \to X$. The intuition is that when the computation is in a state $(q,\sigma)$ of $\aut B$, the value that $\aut A$ would hold in the register $x$ is contained in the register $\sigma(x)$ of $\aut B$.

If $q$ is an initial state of $\aut A$, then $(q, \id_X)$ is an initial state of $\aut B$ with the same initialization. The renaming function $\id_X$ is the identity, which essentially means that the registers are not renamed.

If $q$ is a final state of $\aut A$, then $(q,\sigma)$ is a final state of $\aut B$. The finalization term for $(q,\sigma)$ in $\aut B$ is defined as $\sigma(F(q))$.

Since $\aut A$ is copyless, every register update function $\theta$ induces \emph{at least one} bijection $\upsilon: X \to X$, where $\upsilon(x)$ is the register that appears on the right-hand side of the assignment to $x$ (if such a register exists). If no register is used (i.e., $x$ is set to $\delta(c)$ for a constant $c$), then $\upsilon(x)$ can be set arbitrarily to one of the registers that does not appear in the right-hand side of $\theta$ at all. Thus, there will be multiple possible bijections $\upsilon$, and we choose one arbitrarily. For a transition from state $p$ to state $q$ in $\aut A$ with label $a$ and register update $\theta$, we put for all renaming functions $\sigma$ a transition from $(p,\sigma)$ to $(q,\tau)$ in $\aut B$ with label $a$ and register update $\eta$, where:
\begin{align*}
\tau &= \upsilon; \sigma
&
\eta(x) &= \sigma(\theta(\tau^{-1}(x)))
\end{align*}
and $;$ is function composition (in diagrammatic order). In other words, if $x := \delta(y)$ is a register update of $\aut A$, then we put $\tau(x) := \delta(\sigma(y))$ for the corresponding register update in $\aut B$. Moreover, notice that $\upsilon(x) = y$ and therefore $\tau(x) = (\upsilon; \sigma)(x) = \sigma(\upsilon(x)) = \sigma(y)$. So, the update in $\aut B$ has the same register appearing in both the left-hand and right-hand side.

The main claim is that the possible runs in $\aut B$ correspond to the runs of $\aut A$ modulo the variable renaming described by the renaming functions.
\end{proof}

For an alphabet $\Sigma$ that contains the symbol $\#$, we define the \emph{prefix removal} operation $\partial_\#: \Sigma\kstar \to \Sigma\kstar$ as follows:
\[
  \partial_\#(w) = \begin{cases}
    \eps, &\text{if $\#$ does not appear in $w$}; \\
    v, &\text{if $w = u \# v$ where $u$ has no occurrence of $\#$.}
  \end{cases}
\]
That is, $\partial_\#$ removes the prefix of the string until the first occurrence of a $\#$.

Suppose now that $\Sigma$ and $\Gamma$ are finite alphabets, and that $\Gamma$ contains the $\#$ symbol. For a string transduction of type $\Sigma\kstar \pto \Gamma\kstar$, we define the \emph{output prefix removal} operation $\partial_\#$ as follows:
\[
  \partial_\#(f)(w) = \begin{cases}
    \text{undefined}, &\text{if $f(w)$ is undefined}; \\
    \partial_\#(f(w)), &\text{if $f(w)$ is defined.}
  \end{cases}
\]
We show below that there is a construction on CRAs that corresponds to the output prefix removal operation on string transductions.

\begin{lemma}[Remove Output Prefix]
\label{lemma:outputPrefix}
\normalfont
Suppose that $\Sigma$ and $\Gamma$ are tag alphabets with $\# \in \Gamma$, $\Gamma\kstar$ is the set of values, and $\O$ consists of the constant $\eps$ and the unary operations $(- \cdot a)$ for every $a \in \Gamma$. If the string transduction $f: \Sigma\kstar \pto \Gamma\kstar$ is implementable by a UCRA (resp., DCRA) over $\O$, then $\partial_\#(f)$ is implementable by a UCRA (resp. DCRA) with the same number of registers.
\end{lemma}
\begin{proof}
Suppose that the UCRA $\aut A = (Q,X,\Delta,I,F)$ implements $f$. We will describe a UCRA $\aut B$ that implements $\partial_\#(f)$. The idea is to employ a product construction that records some additional information in the finite control: whether the first $\#$ symbol would have already been added to a register in the execution of $\aut A$ or not. Additionally, we modify the updates of $\aut A$ so that no tags are appended to registers until after the first $\#$ would have been appended in the execution of $\aut A$.

The state space of $\aut B$ is $Q \times (X \to \{0,1\})$. A run of $\aut A$ that ends in a configuration $(q,\alpha) \in Q \times (X \to \Gamma\kstar)$ corresponds to a run of $\aut B$ that ends in a configuration $(q,\rho,\beta) \in Q \times (X \to \{0,1\}) \times (X \to \Gamma\kstar)$ that satisfies
\begin{gather*}
\beta(x) = \partial_\#(\alpha(x))
\quad\text{and}\quad
\rho(x) = \begin{cases}
  0, &\text{if $\alpha(x)$ does not contain \#} \\
  1, &\text{if $\alpha(x)$ contains \#}
\end{cases}
\end{gather*}
for every register $x$.
Notice that the construction that is outlined here preserves determinism.
\end{proof}

\begin{lemma}
\normalfont
\label{lemma:numberOfRegisters}
Suppose that $\Sigma = \{ a, b, \# \}$ is the tag alphabet, $\Sigma\kstar$ is the set of values, and $\O$ consists of the constant $\eps$ and the unary operations $(- \cdot a)$, $(- \cdot b)$, and $(- \cdot \#)$. Let $f: \Sigma\kstar \pto \Sigma\kstar$ be a string transduction, and define the string transduction $g: \Sigma\kstar \pto \Sigma\kstar$ as follows:
\begin{align*}
\dom(g) &= \{a,b\}\kplus \cdot \# \cdot \dom(f)
\\
g(u \# v) &= \begin{cases}
  a^i \cdot \# \cdot f(v)\ \text{where $i = |u|$},
  &\text{if $u$ ends with $a$}
  \\
  b^i \cdot \# \cdot f(v)\ \text{where $i = |u|$},
  &\text{if $u$ ends with $b$}
\end{cases}
\end{align*}
for all $u \in \{a,b\}\kplus$ and $v \in \dom(f)$. If the transduction $g$ can be implemented by a copyless DCRA with $k$ registers, then $f$ can be implemented by a copyless DCRA with at most $k/2$ registers.
\end{lemma}
\begin{proof}
Let $\aut A$ be a DCRA with $k$ registers that implements $g$. By Lemma~\ref{lemma:simplification}, we can assume without loss of generality that every register update of $\aut A$ is of the form $x := u$ or $x := x \cdot u$ for some $u \in \Sigma\kstar$. Consider all register updates that are enabled when the automaton $\aut A$ consumes the input sequences $\{ a, b \}\kstar$. Let $R_a$ (resp., $R_b$) be the set of registers that eventually only contain and are extended with $a$ tags (resp., $b$ tags) in these updates. The registers that are not among $R_a$ and $R_b$ cannot contribute to the output of the computation for sufficiently long input prefixes before the first occurrence of $\#$. One of $R_a$, $R_b$ is of cardinality at most $k/2$. W.l.o.g. we can assume that $|R_a|$ is at most $k/2$. Then, consider the product $\aut B$ of $\aut A$ with an automaton that accepts $a\kplus \cdot \# \cdot \dom(f)$. In the automaton $\aut B$, only the registers of $R_a$ contribute nontrivially to the output, so we can assume that $R_a$ is the register set of $\aut B$. Notice that $\aut B$ computes the transduction $h: \Sigma\kstar \pto \Sigma\kstar$, given by:
\begin{align*}
\dom(h) &= a\kplus \cdot \# \cdot \dom(f)
&
h(a^i \# v) &=
a^i \cdot \# \cdot f(v)
\end{align*}
for all $i \geq 1$ and $v \in \dom(f)$. Now, Lemma~\ref{lemma:outputPrefix} implies that a DCRA $\aut C$ with at most $k/2$ registers can compute the transduction $k: \Sigma\kstar \pto \Sigma\kstar$, given by:
\begin{align*}
\dom(k) &= a\kplus \cdot \# \cdot \dom(f)
&
k(a^i \# v) &= f(v)
\end{align*}
for all $i \geq 1$ and $v \in \dom(f)$. Let $(q,\alpha)$ be the configuration of $\aut C$ after consuming $a \#$. The DCRA that results from $\aut C$ by setting the start state to be $q$ and the initialization function to $\alpha$ computes the transduction $f$ and has at most $k/2$ registers.
\end{proof}

\begin{theorem}
\label{thm:separation}
\normalfont
For some $\O$, there is a string transduction that can be computed by a copyless UCRA over $\O$ but not by a copyless DCRA over $\O$.
\end{theorem}
\begin{proof}
Suppose that $\Sigma = \{ a, b, \# \}$ is the tag alphabet, $\Sigma\kstar$ is the set of values, and $\O$ consists of the constant $\eps$ and the unary operations $(- \cdot a)$, $(- \cdot b)$, and $(- \cdot \#)$. Suppose that $\Sigma = \{ a,b,\# \}$ and $f: \Sigma\kstar \pto \Sigma\kstar$ is the string transduction with domain $\dom(f) = \{a,b\}\kplus \cdot \#$ that is defined as follows:
\[
  f(u \#) = \begin{cases}
    a^i \#, &\text{if $|u| = i$ and $u$ ends with $a$}
    \\
    b^i \#, &\text{if $|u| = i$ and $u$ ends with $b$}
  \end{cases}
\]
for all $u \in \{a,b\}\kplus$. Any copyless DCRA that computes $f$ needs at least two registers. Indeed, if there is only one register, then on any input string it either contains an $a$ or it contains a $b$ or it contains neither; if we take the input string to be sufficiently large and then read in a $b$ or an $a$, there is not enough information stored in the state to know the value of $|u| = i$ and output $b^i$ or $a^i$. Define $f^n$ to be the string transduction with domain $\dom(f^n) = (\{a,b\}\kplus \cdot \#)^n$ and
\[
  f^n(u_1 \# u_2 \# \cdots u_n \#) =
  f(u_1 \#) f(u_2 \#) \cdots f(u_n \#)
\]
where every $u_i$ is an element of $\{a,b\}\kplus$. A consequence of Lemma~\ref{lemma:numberOfRegisters} is that any copyless DCRA that computes $f^n$ requires at least $2^n$ registers.

Now, consider the string transduction $f\kstar$ with domain $(\{a,b\}\kplus \cdot \#)\kstar$ defined as $f\kstar(u_1 \# \cdots u_n \#) = f^n(u_1 \# \cdots u_n \#)$, where $u_i\in\{a,b\}\kplus$ for all $1\leq i \leq n$. The transduction $f\kstar$ can be computed by a copyless UCRA that guesses the last letter of each block (maximal segment of letters that ends with a $\#$ symbol).

Assume for the sake of contradiction that there exists a copyless DCRA $\aut A$ with $m \geq 1$ registers that computes $f\kstar$. The product of $\aut A$ with a DFA that accepts the language $(\{a,b\}\kplus \cdot \#)^k$ has $m$ registers and computes the transduction $f^k$. But we have already discussed that such a DCRA would need at least $2^k$ registers, hence taking $k$ large enough that $2^k > m$ gives us the desired contradiction.
\end{proof}

\subsection{Closure Under Regular Combinators}
\label{subsec:closure}

We will now show that both classes $\SR$ (streamable regular) and $\SLR$ (streamable linear regular) of transductions are closed under operations that are relevant for the design of streaming query languages. In particular, we will consider the collection of the regular combinators of Table~\ref{table:combinators}, which were defined for data transductions.

\begin{theorem}
\normalfont
For every $\O$, the classes of transductions $\SR$ and $\SLR$ (see Definition~\ref{def:streamable} and Table~\ref{table:classesCRA}) are closed under the regular combinators of Table~\ref{table:combinators}.
\label{thm:closure-under-regular-combinators}
\end{theorem}
\begin{proof}
The proof of the theorem relies on constructions that are variants of well-known constructions on classical finite-state automata. We will present the proofs only for the cases $\op(f, g)$ and $\iterQ(f, c, \op)$ and leave the rest of the cases as exercises for the reader. For the cases that we present, we will only give the main ingredients of the construction.

Suppose that $\aut A$ and $\aut B$ are CRAs. We will construct a CRA $\aut C$ that computes $\op(\sem{\aut A}, \sem{\aut B})$, where $\op$ is a binary operation on $D$. The state space of $\aut C$ is the product of the state space of $\aut A$ and the state space of $\aut B$. The registers of $\aut C$ is the union of the registers of $\aut A$ and the registers of $\aut B$. The automaton $\aut C$ simulates the parallel execution of $\aut A$ and $\aut B$, where the registers of the two machines are updated separately. When a state $(p,q)$ of $\aut C$ consists of a final state $p$ of $\aut A$ and a final state $q$ of $\aut B$, the finalization term is defined so that it combines the output values of the two automata $\aut A$ and $\aut B$. This construction preserves the ``copylessness'' property: if the automata $\aut A$ and $\aut B$ are copyless, then so is $\aut C$.

Suppose that $\aut A$ is a CRA, $c$ is a data value, and $\op$ is a binary operation on values. We will construct a CRA $\aut C$ that computes $\iterQ(\sem{\aut A}, c, \op)$. W.l.o.g.\ we assume that we can use $\eps$-transitions, and therefore $\aut A$ has a unique initial state and a unique final state. We will also assume that $\eps \notin \rate(\sem{A})$, because this special case can be handled separately to avoid $\eps$-cycles. We will describe the construction of $\aut C$ as a sequence of modifications on $\aut A$.
\begin{itemize}[$-$]
\item
First, add a new register $y$ to $\aut A$ that is meant to hold the aggregate value computed in the iteration. The register $y$ stays unchanged in the old transitions of $\aut A$.
\item
Add an $\eps$-transition from the unique final state to the unique initial state. In this transition, update the aggregate value of $y$ to reflect one more step of the iteration.
\item
Add a new initial and final state for computing the initial aggregate $c$.
\item
Since the NCRA that has resulted from the previous steps could potentially be \emph{ambiguous}, take the product with a (register-free) DFA that accepts the language $\rate(\sem{A})\kostar$.
\end{itemize}
The last step of the construction removes the ambiguity. So, the resulting CRA is unambiguous and computes the transduction $\iterQ(\sem{\aut A}, c, \op)$. This construction preserves the property of ``copylessness'': if $\aut A$ is copyless then so is $\aut C$.
\end{proof}

In contrast to Theorem~\ref{thm:closure-under-regular-combinators}, we can show that the class of \emph{all streaming algorithms} with an appropriate streaming complexity bound is not closed under the regular combinators of Table~\ref{table:combinators}. For instance, the following theorem shows that this is true if by ``streaming algorithm'' we require that the algorithm must use only $O(\log(n))$ space in the length of the input stream. This justifies the role of regularity in guaranteeing modular composition, as we mentioned in the introduction.

\begin{theorem}
\normalfont
Let $\Sigma = \{a, b\}$ and $D = \mathbb{N}$. There exist data transductions $f$ and $g$ (\emph{not} in $\SR$) such that $f$ and $g$ can be implemented as streaming algorithms which use at most $O(\log(n))$ space in the length of the input stream ($n$), but $\splitQ(f,g,\op)$ cannot.
\label{thm:general-streaming-algorithms-not-closed}
\end{theorem}
\begin{proof}[Proof]
This example is adapted from Theorem 5.3 in \cite{AMS2019}.
Define $f, g: (\Sigma \times \N)\kstar \pto \N$ as follows: $f(w) = 1$ if $|w|_{a} = 2 \cdot |w|_{b}$, and undefined otherwise, and $g(w) = 1$ if $|w|_{a} = |w|_{b}$, and undefined otherwise. Here $|w|_{a}$ is the number of occurrences of $a$ in $w$.
Note that $f$ and $g$ are value-oblivious, and neither $\rate(f)$ nor $\rate(g)$ is regular. The rate of the transduction $\splitQ(f,g,\op)$ (the choice of $\op$ doesn't matter) is the concatenation of these two languages, which is an unambiguous concatenation.

Both $f$ and $g$ can be implemented by maintaining two counters for the number of $a$'s and the number of $b$'s seen so far.
This requires $O(\log n)$ space.
On the other hand, any streaming algorithm that computes $h = \splitQ(f,g,\op)$ requires at least a linear number of bits. Specifically, consider the behavior of such a streaming algorithm on inputs of the form $a(aab | aba)^n ab$. On these $2^n$ distinct inputs, each of length $3n+3$, the streaming algorithm would have to reach $2^n$ different internal states, because the inputs are pairwise distinguished by reading in a further string of the form $b^k$.
For example, the inputs $w = a(aab)ab$ and $w' = a(aba)ab$ are distinguished by reading in $b$, and in general if we have two inputs
\begin{align*}
w &\in a(aab | aba)^{n-k} (aab) (aab | aba)^{k-1} ab \\
\text{and}\quad w' &\in a(aab | aba)^{n-k} (aba) (aab | aba)^{k-1} ab,
\end{align*}
they are distinguished by reading in the string $b^k$.

Thus on inputs of size $3n + 3 = \Theta(n)$ the streaming algorithm requires at least $n$ bits to store the state, which is not $O(\log n)$.
\end{proof}

Finally, we note as another contrast to Theorem~\ref{thm:closure-under-regular-combinators} that there are interesting and natural operations on data transductions under which $\SR$ is closed, but $\SLR$ is not closed. In particular, let $f: (\Sigma \times D)\kstar \pto D$ be a data transduction such that $\rate(f) = \Sigma\kstar$, i.e. $f$ is \emph{total}, let $c$ be a constant in $\O$ and $\op$ a binary operation in $\O$, and define the \emph{prefix sum} of $f$ by:
\begin{align*}
\prefsumQ(f, c, \op) &: (\Sigma \times D)\kstar \pto D
\\
\rate(\prefsumQ(f, c, op)) &= \Sigma\kstar
\\
\prefsumQ(f, c, \op)(w) &=
\fold(\op, c, [f(\eps), f(w_1), \ldots, f(w_1 w_2 \ldots w_n)]),
\end{align*}
where $w = w_1, \ldots, w_n$ and each $w_i \in \Sigma \times D$. Then:

\begin{theorem}
\normalfont
For every $\O$, the class of data transductions $\SR(\O)$ is closed under the operation $\prefsumQ$. For some $\O$, the class $\SLR(\O)$ is not closed under $\prefsumQ$.
\label{thm:prefix-sum-closure}
\end{theorem}
\begin{proof}
First we show $\SR$ is closed. Let $c \in \O$ be a constant, $\op \in \O$ a binary operation, and $\aut A = (Q,X,\Delta,I,F)$ a copyful UCRA which denotes a total transduction. By Theorem~\ref{thm:UCRAtoDCRA}, assume w.l.o.g. that $\aut A$ is deterministic; additionally, assume it is trim. Then, we construct a DCRA for $\prefsumQ(f, c, \op)$ by adding one register $\mathsf{total}$, and updating it along every transition by applying $\op$ to the previous $\mathsf{total}$ and the new output of $\aut{A}$. We always know the new output of $\aut{A}$, because in a deterministic, trim automaton which denotes a total transduction, every state is final.

Next we show that $\SLR$ is \emph{not} closed. Let $\Sigma = \{a\}$ be the tag alphabet, let $M = (D, \cdot, \eps)$ be the free monoid generated by $\Sigma$, and consider the family of operations $\O_M$. Let $f: (\Sigma \times D)\kstar \pto D$ be the transduction which outputs the same number of $a$'s in the input: $f((a, d_1) \ldots (a, d_n)) = a^n$. Then $f \in \SLR(\O)$ because it can be implemented by a copyless UCRA. However, $\prefsumQ(f, \eps, \cdot)$ has the property that on inputs of length $n$, the output $a^{n(n+1)/2}$ is of quadratic size. As we have observed in the comments following Lemma~\ref{lemma:monoidWeak}, the output of a copyless transduction over $\O_M$ must be linear in the input, so this means that $\prefsumQ(f, \eps, \cdot) \notin \SLR(\O)$.
\end{proof}

\subsection{Complexity of Evaluation}
\label{subsec:complexity}

We are interested in evaluation in the streaming model of computation. The main resources in this model are: 1) the total space required by the algorithm, and 2) the time needed to process each data item. Both these resources are given as a function of the length of the stream that has been seen so far. Ideally, both these key parameters should be constant or logarithmic in the length of the stream.

Suppose that $\aut A = (Q,X,\Delta,I,F)$ is a trim UCRA. The computation of the semantics of $\aut A$ in the streaming model of computation is similar to the deterministic simulation of classical NFAs, in which the set of \emph{active states} is maintained and updated at every step of the computation. The difference for CRAs is that we need to maintain for every active state $q$ a register assignment $\alpha_q: X \to D$. Unambiguity guarantees that at every step of the computation we need to store \emph{at most one} register assignment for every state. This property does not hold for ambiguous machines, as there are several different paths that could lead to the same state. So, the total number of values that are needed to store is bounded above by $|Q| \cdot |X|$. For deterministic CRAs, the corresponding upper bound is $|X|$.
Analogous upper bounds can be given on the \emph{number of operations} in $\O$ that need to be executed (roughly, the time needed) per element of the input stream. For both copyless and copyful cases, these bounds are the same.

However, a more precise accounting of the used computational resources requires looking at the number of \emph{bits} that are needed to store each value during the computation (rather than just the number of values that need to be stored and the number of operations that need to be computed). In order to make such finer distinctions, we need to take into account the operations that are allowed in $\O$. For example, suppose that the values are the natural numbers and that $\O$ contains numerical constants, binary addition, $\max$ and $\min$. In the copyless case, the contents of registers can grow only linearly in the length of the input and therefore a logarithmic numbers of bits is required for each register. In the copyful case, however, even if we restrict $\O$ to contain only $0$ and binary addition, the contents of registers can grow exponentially and therefore a linear number of bits may be needed.

In some applications, the data values range over a finite domain (e.g., 32-bit integers, 64-bit floating-point numbers, and so on), and therefore the space needed to store each register is constant. Moreover, the operations on such bounded values require constant time. In this setting, the space and time-per-element resource requirements for evaluating CRAs is constant.

\section{MSO-Definable Transductions}
\label{sec:MSO}

In this section, we give logical characterizations of the classes of transductions defined by the Cost Register Automata of Definition~\ref{def:CRA}:
\begin{enumerate}[(1)]
\item
streamable regular or $\SR$ (computed by copyful UCRAs),
\item
streamable linear regular or $\SLR$ (computed by copyless UCRAs).
\end{enumerate}
These logical characterizations are in terms of word-to-DAG transformations which are given by formulas written in the language of monadic second-order logic (MSO).

The MSO-definable transductions that we consider here specify output graphs \emph{with no backward edges}. This restriction for the output DAGs to only have forward edges is crucial so that the function is implementable by an automaton which maintains values (as opposed to maintaining terms which may have ``holes'', as in \cite{AdA2012STT}). Indeed, we will show that the forward-only MSO-definable transductions are exactly the transformations that are implementable by the value-based UCRAs of Definition~\ref{def:CRA}.

This class of transductions has previously arisen in the study of attribute grammars \cite{BE2000} under the name of \emph{direction preserving} MSO graph transductions (page 11 in \cite{BE2000}). An important difference here compared to \cite{BE2000} is the characterization of the transduction classes in terms of automata (\cite{BE2000} uses attribute grammars) whose efficient evaluation in the streaming model of computation is obvious. Another technical difference is that the transductions in the present paper involve data values, and we use a special symbol $\val$ and the position of a $\val$-labeled vertex in order to define the semantics of an MSO transduction. Our proof is self-contained, and exhibits a pleasant uniformity due to the translation from forward-only MSO transductions to \emph{unambiguous} CRAs. Recall that unambiguity is essential in order to capture the class $\SLR$ in the copyless CRA model (Theorem~\ref{thm:separation}). In the word-to-DAG case, the output graphs can have more than one edge emanating from a vertex, which corresponds to copying a register in the machine model. In the word-to-tree case, each vertex of the output graph has at most one edge emanating from it, which implies that registers are used in a copyless manner in the machine model.

An \emph{MSO transduction} specifies the output DAG using MSO formulas for the vertices and edges that are interpreted over the entire input stream. The challenge is to convert this global description of the output into local computation rules for the automaton. In particular, the existence of a vertex or edge depends on a regular property of the \emph{future} of the input (regular look-ahead). To deal with this regular look-ahead, we make use of \emph{unambiguously nondeterministic} CRAs. The unambiguous nondeterminism conveniently allows us to make guesses regarding a regular property of the future of the input.

The crucial intermediate result for obtaining equivalence between MSO transductions and CRAs (Corollary~\ref{coro:MSOtoCRA} and Theorem~\ref{thm:CRAtoMSO}) asserts that every MSO transduction can be transformed equivalently into a \emph{single-step} MSO transduction, namely one where the edges of the output DAG extend over at most one input element (Lemma~\ref{lemma:singleStep}). To implement a single-step MSO transduction it is then sufficient to maintain as many data registers as the maximum number of DAG vertices that appear in some position. As before, w.l.o.g.\ we allow the use of $\eps$-transitions for unambiguous CRAs.

Let $\Sigma$ be the finite alphabet of tags, $D$ a (possibly infinite) set of data values, and $\O$ a collection of constants and operations on $D$. Recall that we only consider transductions $f: (\Sigma \times D)\kstar \pto D$ that satisfy the following condition: whether $f(w)$ is defined or not depends only on the sequence of tags $w |_\Sigma$.

\subsection{MSO Transductions}

An MSO-definable graph transduction \cite{C1994MSOGT} specifies a partial function from labeled graphs to labeled graphs. The nodes, edges and labels of the output graph are described in terms of MSO formulas over the nodes, edges and labels of the input graph.

The tag projection of a data word $w = (a_1,d_1) \sspc (a_2,d_2) \ldots (a_n,d_n)$ of length $n$ over $\Sigma \times D$ can be represented as a labeled directed graph $\G(w)$ with $n+1$ vertices $\{ 0,1,\ldots,n \}$ and $n$ $\Sigma$-labeled edges $(i-1,a_i,i)$ for all $i=1,\ldots,n$. For example, we visualize the graph $\G(w)$ of the word $w = (\tg A,1) \sspc (\tg B,2) \sspc (\tg A,3) \sspc (\tg B,4) \sspc (\tg A,1)$ over the alphabet with tags $\Sigma = \{ \tg{A}, \tg{B} \}$ and values that are natural numbers as follows:
\[
  \includegraphics{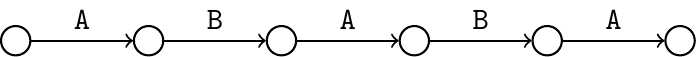}
\]
Note that the values in the input word are not represented in the graph $\G(w)$; they will be represented later using the special symbol $\val$.

A graph $\G(w)$ specifies a relational structure whose domain is the set of vertices and which has a binary predicate $\edge^w_\tg{a}$ for every tag $\tg{a} \in \Sigma$ consisting of the $\tg{a}$-labelled edges.
We use the \emph{monadic second-order language} $\MSO(\Sigma)$ to express properties of words when represented as graphs. For every tag $\tg{a}$ of the alphabet $\Sigma$, the language has a binary predicate symbol $\edge_\tg{a}$. Informally, $\edge_\tg{a}(x,y)$ means that there is an edge from $x$ to $y$ labeled with the tag $\tg{a}$. The language also contains the equality predicate $=$ (i.e., $x=y$ means that the vertices $x$ and $y$ are equal) and the containment predicate $\in$ (i.e., $x \in X$ means that the vertex $x$ belongs to the set of vertices $X$). The formulas of $\MSO(\Sigma)$ are built from atomic formulas using the Boolean connectives, first-order quantification (over vertices of $\G(w)$), and second-order monadic quantification (over sets of vertices of $\G(w)$). We do not need to include $\le$, because we will shortly introduce it as an abbreviation.

Note that we use binary predicates $\edge_{\tg{a}}(x,y)$ to encode tags, instead of the more commonly used unary predicates \cite{T1997LAL}. That is, $\G(w)$ has labeled edges instead of labeled vertices. The reason for this encoding is that we want to allow transductions that map the empty input word into some output graph (as is done, for example, in page 232 of \cite{EH2001}). In any case, this difference is inconsequential for the definable class of functions.

We write $\Word(\Sigma)$ for the class of graph structures over the signature of $\MSO(\Sigma)$ that correspond to finite words over $\Sigma$. If $\phi$ is a sentence of $\MSO(\Sigma)$ and $w$ is a word over $\Sigma$, then we write $w \models \phi$ to denote that the graph $\G(w)$ satisfies $\phi$. The abbreviation $\edge(x,y) := \tbigvee_{\tg{a} \in \Sigma} \edge_{\tg{a}}(x,y)$ says there is an edge from $x$ to $y$. Moreover, we encode $\le$ as follows:
\[
(x \le y) := \forall X. x \in X \land (\forall u \forall v. u \in X \land \edge(u,v) \to v \in X) \to y \in X.
\]

A \emph{ranked alphabet} $\Gamma$ is a finite set with assignment of a natural number to each element, called its \emph{arity}. A \emph{directed acyclic graph} (DAG) over a set of ranked symbols $\Gamma$ has $\Gamma$-labeled vertices, a single sink vertex (its \emph{root}), and edges indexed between $1$ and the maximum arity, such that the graph respects the arity of the symbols of $\Gamma$. For example, suppose $\Gamma = \{ a,b,f \}$, where $a$ and $b$ are constants (arity 0) and $f$ is binary (arity 2). The following graph is a DAG over $\Gamma$, where we use arrows $\to$ for the first argument of $f$, and arrows $\tto$ for the second argument of $f$:
\[
  \includegraphics{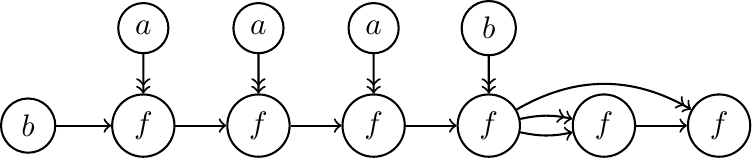}
\]

\begin{definition}[Forward-only word-to-DAG MSO transductions]
\label{def:msoTransduction}
\normalfont
Let $\Gamma = \O \cup \{ \val \}$ be the ranked alphabet that includes the operations $\O$ on the data values $D$ as well as a fresh constant symbol $\val$.
Let $\imax$ denote the maximum arity of a symbol in $\Gamma$.
A \emph{forward-only MSO-definable word-to-DAG transduction} over $\O$ from $(\Sigma \times D)\kstar$ to $D$ consists of:
\begin{enumerate}
\item
An $\MSO(\Sigma)$ sentence $\phi_\dom$, called the \emph{domain sentence}.
\item
A finite \emph{copy set} $C$.
\item
For every copy element $c \in C$ and every symbol $\gamma \in \Gamma$, an $\MSO(\Sigma)$ \emph{vertex formula} $\phi^c_\gamma(x)$ with one free variable $x$.
\item
For all copy elements $c,d \in C$ and for all argument indices $i$ from $1$ to $\imax$, an $\MSO(\Sigma)$ \emph{edge formula} (for the $i$-th argument) $\psi_i^{c,d}(x,y)$ with two free variables $x$ and $y$ s.t. $\models \forall x \forall y.\, \psi_i^{c,d}(x,y) \to x \leq y$.
\end{enumerate}
\emph{such that} all of the well-formedness conditions shown in Observation~\ref{obs:wellformedness} hold.
These conditions assert that the formulas define a properly structured DAG over the symbols $\Gamma$ 
in the output structure (which consists of $C$ copies of the vertices of $\G(w)$).

\textbf{\em Semantics}:
The MSO transduction \emph{denotes} a partial function $\tau: (\Sigma \times D)\kstar \pto D$ whose domain is the set of models $\sem{\phi_\dom}$ of the domain sentence. For every data word $w$ satisfying $\sem{\phi_\dom}$, the transduction specifies a single DAG over the symbols $\O \cup \{ \val \}$. Every symbol in $\O$ is interpreted as an operation on $D$, and $\val$ is a special symbol for the data value of the input string at the corresponding position. The DAG therefore can be evaluated to a value in $D$, which is defined to be the value of $\tau$ on $w$.
\end{definition}

Note that while the formulas above \emph{define} an output graph, they are only logically \emph{interpreted over} an input graph $\G(w)$. This is why they define a single output structure (not, e.g., a set of output structures or a relation between input and output structures). The well-formedness conditions of Observation~\ref{obs:wellformedness} then guarantee this output structure has the correct form. These conditions are stated as MSO formulas, but it is not actually necessary for our results that they are expressible in MSO; rather, we only use the fact that the output of every forward-only MSO transduction is a properly-structured DAG, and we have baked this condition into the definition of such transductions. However, the fact that they are $\MSO(\Sigma)$ formulas means the conditions are effectively checkable (decidable).

\begin{observation}[Well-formedness]
\label{obs:wellformedness}
\normalfont
We list here the well-formedness conditions for the forward-only MSO-definable transductions of Definition~\ref{def:msoTransduction}. Note that all conjunctions and disjunctions ($\tbigvee,\tbigwedge$) are finite, and there are finitely many formulas overall.
Logical validity is over all structures in $\Word(\Sigma)$. I.e., for every input graph $\G(w) \in \Word(\Sigma)$ (and all interpretations of the free variables), the formula should be true.
\begin{enumerate}[(1)]
\addtolength{\itemsep}{-1ex}
\item
\textbf{At most one vertex label}:
$\tbigwedge_{\gamma \in \Gamma} \tbigwedge_{\delta \in \Gamma, \delta\neq\gamma} \neg\phi^c_\gamma(x) \lor \neg\phi^c_\delta(x)$
is valid for every $c \in C$.
\item
\textbf{Edges connect active vertices}: $\psi_i^{c,d}(x,y) \to \phi^c(x) \land \phi^d(y)$ is valid for all $c,d \in C$ and $i=1,2,\ldots,\imax$, where $\phi^c(x) := \tbigvee_{\gamma \in \Gamma} \phi^c_\gamma(x)$ abbreviates that the $c$-vertex at position $x$ is active. Informally, a vertex $(c,x)$ of the output graph (where $c$ is an element of the copy set $C$ and $x$ is a position) is \emph{active} if it is labeled with some symbol, i.e. $\phi^c(x)$ is true (when interpreted in the input graph).
\item
\textbf{(Non)existence of output}:
The formulas
$\phi_\dom \to \exists x \tbigvee_{c \in C} \phi^c(x)$ and $\neg\phi_\dom \to \forall x \tbigwedge_{c \in C} \neg\phi^c(x)$ are valid.$\val$
\item
\textbf{No local cycle at any position}: For every potential cycle $c_1 \to^{i_1} c_2 \to^{i_2} \cdots \to^{i_{k-1}} c_k \to^{i_k} c_1$ of length $k \leq |C|$, where $c_j \in C$ and $1 \le i_j \le \imax$ are edge labels, the following formula is valid:
\[
  \lnot \, \exists x.\,
  \psi_{i_1}^{c_1,c_2}(x,x) \land
  \psi_{i_2}^{c_2,c_3}(x,x) \land
  \cdots \land
  \psi_{i_k}^{c_k,c_1}(x,x).
\]
\item
\textbf{The arity of labels is respected}: For every symbol $\gamma \in \Gamma$ of arity $n$, every $i=1,\ldots,n$, and every $c \in C$, the following formula is valid:
\begin{align*}
\phi^d_\gamma(y) \to \exists x.\, \Bigl( &
\tbigvee_c(\psi_i^{c,d}(x,y) \land \tbigwedge_{\hat c \neq c} \neg\psi_i^{\hat c,d}(x,y)) \land {}
\\[-0.5ex]
&\forall \hat x \tbigwedge_{\hat c} (
   \psi_i^{\hat c,d}(\hat x,y) \to \hat x = x)
\Bigr).
\end{align*}
\item
\textbf{Global uniqueness of sink}: For all $c \in C$, the formula
\[
  \sink_c(x) \to
  \tbigwedge_{d \neq c}\neg\sink_c(x) \land
  \forall y.(y \neq x \to \tbigwedge_d \neg\sink_d(y))
\]
is valid, where $\sink_c(x) := \forall y \tbigwedge_d \tbigwedge_i \neg\psi^{c,d}_i(x,y)$ abbreviates that there is no outgoing edge from copy $c$ at position $x$.
\item
\textbf{No value at first position}:
$\first(x) \to \tbigwedge_c \neg\phi^c_\val(x)$ is valid,
where $\first(x) := \forall y.\, x \le y$ abbreviates the very first position. The constant symbol $\val$, which is meant to refer to the input data value immediately before the point of occurrence, should not appear in this position.
\end{enumerate}
\end{observation}

\begin{example}[MSO Transduction]
\normalfont
\label{ex:transduction1}
The tag alphabet is $\Sigma = \{ \tg{A}, \tg{B} \}$ and the data domain $D$ is the set of natural numbers. The collection $\O$ of operations on $D$ contains the constant 0 and addition.
\[
  \includegraphics{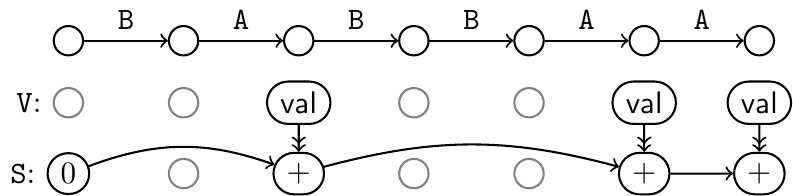}
\]
We will describe the transduction that sums up all data values that are tagged with $\tg{A}$. The copy set is taken to be $C = \{ \tg{V}, \tg{S} \}$. The domain sentence is $\phi_\dom = \true$ and the vertex formulas are:
\begin{align*}
\phi^\tg{V}_\val(x) &=
\phi^\tg{S}_+(x) =
\edge_\tg{A}(x{-}1,x)
&
\phi^\tg{S}_0(x) &= \first(x)
\end{align*}
where $\edge_\tg{A}(x-1,x)$ is notation for $\exists y. \edge_\tg{A}(y,x)$. The edge formulas are given as follows:
\begin{align*}
\psi^{\tg V, \tg S}_2(x,y) = {}
&(x=y) \land \edge_\tg{A}(x{-}1,x)
\\
\psi^{\tg S, \tg S}_1(x,y) = {}
&(\first(x) \lor \edge_\tg{A}(x{-}1,x) ) \land 
\edge_\tg{A}(y{-}1,y) \land {}
\\
&\forall z.\, x<z<y \to \neg\edge_\tg{A}(z{-}1,z)
\end{align*}
All omitted vertex and edge formulas are equal to $\false$.
\end{example}

As we will see later, the data transductions that are definable by forward-only word-to-DAG MSO-transductions are exactly the streamable regular transductions. Similarly, the data transductions that are definable by forward-only word-to-tree MSO-transductions will be shown to be exactly the streamable linear regular transductions. See Table~\ref{table:classesMSO}. The condition that the output graph is a tree can be guaranteed by requiring that every vertex of the graph has at most one outgoing edge. This condition can be expressed by an MSO formula, and therefore its validity can be effectively checked.

\begin{table}[t]
\centering
\begin{tabular}{c|c}
Output graph & Class of transductions
\\ \hline
Tree
& Streamable Linear Regular
\\
DAG
& Streamable Regular
\end{tabular}
\caption{Classes of MSO-definable transductions with forward-only edges.}
\label{table:classesMSO}
\end{table}

\begin{definition}[Models of Open Formulas]
\normalfont
Let $\phi(x)$ be an $\MSO(\Sigma)$ formula with one free individual variable. The set of \emph{models} of $\phi(x)$, denoted $\sem{\phi(x)}$,
\[
  \includegraphics{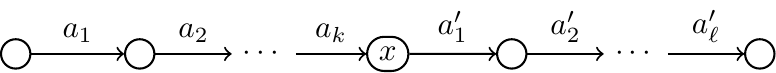}
\]
is defined to be the set of pairs $(u,v) \in \Sigma\kstar \times \Sigma\kstar$ such that the word $u \cdot v$ satisfies $\phi(x)$ under the variable assignment $x \mapsto |u|$.

Let $\psi(x,y)$ be an $\MSO(\Sigma)$ formula with two free individual variables such that $\models \psi(x,y) \to x \leq y$. The set of \emph{models} of $\psi(x,y)$, denoted $\sem{\psi(x,y)}$,
\[
  \includegraphics{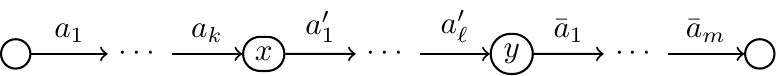}
\]
is defined to be the set of triples $(u,v,w) \in \Sigma\kstar \times \Sigma\kstar \times \Sigma\kstar$ such that the word $u \cdot v \cdot w$ satisfies $\psi(x,y)$ under the variable assignment $x \mapsto |u|$, $y \mapsto |uv|$.
\end{definition}

\begin{lemma}[\cite{EH2001MDST}, Lemma 5]
\label{lemma:MSOtoRegex}
\normalfont
Let $\phi(x)$ be an $\MSO(\Sigma)$ formula with one free individual variable. The set of models $\sem{\phi(x)}$ is a finite disjoint union of sets of the form $\sem{r_1} \times \sem{r_2}$, where $r_1$ and $r_2$ are regular expressions over $\Sigma$.

Similarly, suppose that $\psi(x,y)$ is an $\MSO(\Sigma)$ formula with two free individual variables such that $\models \psi(x,y) \to x \leq y$. The set of models $\sem{\psi(x,y)}$ is a finite disjoint union of sets of the form $\sem{r_1} \times \sem{r_2} \times \sem{r_3}$, where $r_1$, $r_2$ and $r_3$ are regular expressions over $\Sigma$. \end{lemma}
\begin{proof}[Proof sketch]
This is a folklore result for regular languages (see, for example, Lemma 5 of \cite{EH2001MDST}). For the first part of the lemma, the idea is to construct the DFA $\aut A$ that accepts the models of $\phi(x)$. The annotation with $x$ is performed by the automaton by allowing some states to be labeled with $x$. For every state $q$ of $\aut A$ that is labeled with $x$, we include the regular expression for the set $L_q \cdot R_q$, where $L_q$ is the set of strings that lead from the initial state to $q$, and $R_q$ is the set of strings that lead from $q$ to a final state of $\aut A$. The construction for the second part of the lemma is similar.
\end{proof}

Using Lemma~\ref{lemma:MSOtoRegex} an MSO transduction can be transformed into a set of rules that are expressed equivalently with regular expressions instead of MSO formulas. The domain sentence $\phi_\dom$ can be given as a regular expression. A vertex formula $\phi^c_\gamma(x)$ can be given as a finite set of pairwise disjoint \emph{vertex rules} of the form
\[
  \text{label $\gamma$ at $c$}:\ r_1\sepa r_2,
\]
where $r_1$ and $r_2$ are regular expressions. The rule specifies a prefix (from start to $x$) with $r_1$ and a suffix (from $x$ to end) with $r_2$ so that the copy $c$ of the output DAG is active at position $x$. Similarly, an edge formula $\psi^{c,d}_i(x,y)$ of the transduction can be presented as a finite set of pairwise disjoint \emph{edge rules} of the forms
\begin{align*}
\text{edge $c \to_i d$} &:\ r_1\sepa r_2\sepa r_3
&&\text{[forward edges]}
\end{align*}
where $r_1$, $r_2$ and $r_3$ are regular expressions. The rule specifies a prefix (from start to $x$) with $r_1$, a middle part (from $x$ to $y$) with $r_2$, and a suffix (from $y$ to end) with $r_3$ so that the output DAG has an edge from copy $c$ at position $x$ to copy $d$ at position $y$. We have required in the definition that edges be forward-only, so there are no models in which $y$ comes before $x$ in $\psi^{c,d}_i(x,y)$.

\begin{example}
\label{ex:transduction1regex}
\normalfont
Using Lemma~\ref{lemma:MSOtoRegex} we will present the MSO transduction of Example~\ref{ex:transduction1} equivalently with regular expressions. We have: $\phi_\dom = \Sigma\kstar$ and
\begin{align*}
\phi^\tg{V}_\val =
\phi^\tg{S}_+ &=
(\Sigma\kstar \tg{A})\sepa \Sigma\kstar
&
\psi^{\tg V, \tg S}_2 &=
(\Sigma\kstar \tg{A})\sepa \eps\sepa \Sigma\kstar
\\
\phi^\tg{S}_0 &= \eps\sepa \Sigma\kstar
&
\psi^{\tg S, \tg S}_1 &=
(\eps \cup \Sigma\kstar \tg{A})\sepa
(\tg{B}\kstar \tg{A})\sepa \Sigma\kstar
\end{align*}
The vertex rule $\phi^{\tg S}_+$, for example, says that the copy $\tg S$ is active with label $+$ at every position immediately after the occurrence of a letter $\tg A$ in the input string. Note that here, each formula becomes a single rule; in general, each formula will become a set of pairwise disjoint rules.
\end{example}

A MSO transduction is said to be \textbf{\em single-step} if the edges that it specifies are restricted so that they refer to positions that are at most 1 apart, that is: the formula
\[
  \psi_i^{c,d}(x,y) \to (x=y) \lor \edge(x,y)
\]
is valid for all $c,d \in C$ and every argument index $i$.

\begin{lemma}[Single-Step Decomposition]
\label{lemma:singleStep}
\normalfont
Let $\Sigma$ be a finite tag alphabet, and $D$ be a set of data values with operations $\O$. A forward-only MSO-definable word-to-DAG transduction over $\Sigma, D, \O$ can be transformed into an equivalent single-step transduction over $\Sigma, D, \O \cup \{ \id \}$, where $\id: D \to D$ is the identity function on $D$.
\end{lemma}
\begin{proof}
Suppose that the transduction is presented in the form using regular expressions instead of MSO formulas, as described earlier. We will show how to eliminate rules that violate the single-step constraint by replacing them with rules that are either $\eps$-transitions or single-letter transitions.

Consider a rule $\psi^{c,d}_i = c \to_i d: r_1\sepa r_2\sepa r_3$ for $i$-edges from copy $c$ to copy $d$ which is not an $\eps$-transition, i.e. $r_2 \neq \eps$. The idea is to decompose an edge specified by the rule, which spans a substring matching $r_2$, by simulating the execution of an automaton for the expression $r_2$. Take the minimal DFA for $r_2$ and prune it into an equivalent partial DFA so that every state has a path to some final state. Suppose this automaton is $(Q,\Sigma,\Delta,q_0,F)$. For a state $q \in Q$ define the regular set
\[
  L_q =
  \{ w \in \Sigma\kstar \mid \Delta(q_0,w) = q \},
\]
that is, $L_q$ is the set of strings over $\Sigma$ that lead from the initial state $q_0$ to $q$. Moreover, define the regular set
\[
  R_q =
  \{ w \in \Sigma\kstar \mid \Delta(q,w) \in F \},
\]
that is, $R_q$ is the set of strings over $\Sigma$ that lead from $q$ to some final state. Notice that $L_q \cdot R_q \subseteq r_2$ for every state $q$, and additionally $r_2 = \tbigcup\{ L_q \mid q \in F \}$.

For every state $q$ of the automaton, we introduce a fresh copy element $v_q$. We can replace the rule $\psi^{c,d}_i$ by the following set of vertex and edge rules:
\begin{align*}
\text{label $\id$ at $v_q$} &:
(r_1 \cdot L_q)\sepa (R_q \cdot r_3)
\\
\text{edge $c \to_1 v_I$} &:
r_1\sepa \eps\sepa (r_2 \cdot r_3)
\\
\text{edge $v_q \to_i d$} &:
(r_1 \cdot L_q)\sepa \eps\sepa r_3,
\ \text{for every $q \in F$}
\\
\text{edge $v_p \to_1 v_q$} &:
(r_1 \cdot L_p)\sepa \tg{a}\sepa (R_q \cdot r_3),
\ \text{ for every } (p,\tg{a},q) \in \Delta.
\end{align*}
Since the automaton is deterministic, we are guaranteed that all of these rules are disjoint, and moreover that the modified transduction satisfies all the well-formedness conditions. In addition, it preserves the property of the output DAG being a tree.

The construction is iterated until there is no edge rule left that violates the single-step condition. The resulting transduction can then be equivalently presented using MSO formulas.
\end{proof}

\begin{example}[Single-step]
\label{ex:transduction1singleStep}
\normalfont
We will apply now the construction of the proof of Lemma~\ref{lemma:singleStep} to the transduction of Example~\ref{ex:transduction1}. Using the presentation of Example~\ref{ex:transduction1regex} we see that the only edge rule that violates the single-step condition is $\psi^{\tg S, \tg S}_1$. Now, the pruned partial DFA for $\tg{B}\kstar \tg{A}$ is the following:
\begin{center}
\text{\includegraphics{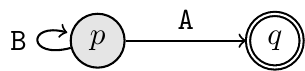}}
\qquad
$\begin{aligned}[b]
L_p &= \tg{B}\kstar
&
R_p &= \tg{B}\kstar \tg{A}
\\
L_q &= \tg{B}\kstar \tg{A}
&
R_q &= \eps
\end{aligned}$
\end{center}
The initial state is indicated with a gray background. We introduce fresh copy elements $p$ and $q$, one for every state of the automaton for $\tg{B}\kstar \tg{A}$. So, the new copy set is $C' = \{ \tg{V}, \tg{S}, p, q \}$. The following figure illustrates how the edges defined by $\psi^{\tg S, \tg S}_1$ are decomposed into $\eps$-transitions and single-letter transitions.
\[
  \includegraphics{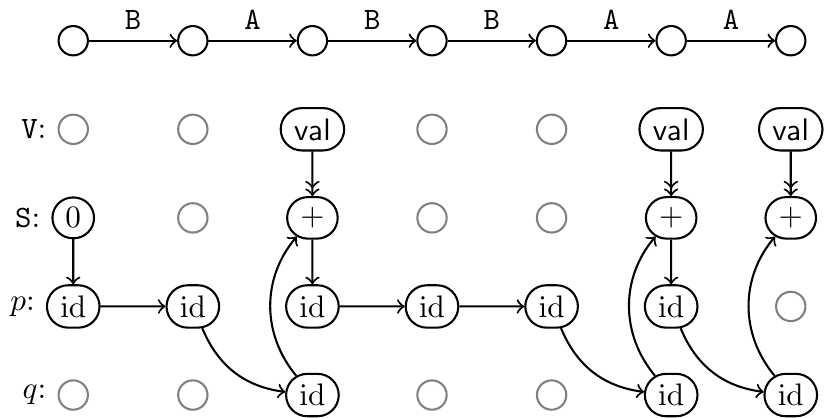}
\]
So, the transduction is specified by the following rules:
\begin{allowdisplaybreaks}
\begin{align*}
\begin{aligned}[t]
\phi_\dom &= \Sigma\kstar
\\
\phi^\tg{V}_\val &=
(\Sigma\kstar \tg{A})\sepa \Sigma\kstar
\\
\phi^\tg{S}_0 &= \eps\sepa \Sigma\kstar
\\
\phi^\tg{S}_+ &= (\Sigma\kstar \tg{A})\sepa \Sigma\kstar
\\
\phi^{p}_\id &=
(\eps \cup \Sigma\kstar\tg{A})\tg{B}\kstar\sepa
\tg{B}\kstar\tg{A}\Sigma\kstar
\\
\phi^{q}_\id &=
(\eps \cup \Sigma\kstar\tg{A})\tg{B}\kstar\tg{A}\sepa
\Sigma\kstar
\end{aligned}
&&
\begin{aligned}[t]
\psi^{\tg V, \tg S}_2 &=
(\Sigma\kstar \tg{A})\sepa \eps\sepa \Sigma\kstar
\\
\psi^{\tg S, p}_1 &=
(\eps \cup \Sigma\kstar\tg{A})\sepa
\eps\sepa
\tg{B}\kstar\tg{A}\Sigma\kstar
\\
\psi^{p, p}_1 &=
(\eps \cup \Sigma\kstar\tg{A})\tg{B}\kstar\sepa
\tg{B}\sepa
\tg{B}\kstar\tg{A}\Sigma\kstar
\\
\psi^{p, q}_1 &=
(\eps \cup \Sigma\kstar\tg{A})\tg{B}\kstar\sepa
\tg{A}\sepa
\Sigma\kstar
\\
\psi^{q, \tg S}_1 &=
(\eps \cup \Sigma\kstar\tg{A})\tg{B}\kstar\tg{A}\sepa
\eps\sepa
\Sigma\kstar
\end{aligned}
\end{align*}
\end{allowdisplaybreaks}%
The $\eps$-transitions $\psi^{\tg S, p}_1$, $\psi^{q, \tg S}$ and the single-letter transitions $\psi^{p,p}_1$, $\psi^{q,\tg S}_1$ expand the edges of $\psi^{\tg S, \tg S}_1$ according to the automaton for $\tg{B}\kstar\tg{A}$.
\end{example}

\begin{theorem}
\label{thm:MSOtoCRA}
\normalfont
For any $\O$, every single-step word-to-dag (resp., word-to-tree) forward-only MSO transduction over $\O$ can be implemented by a copyful (resp., copyless) UCRA over $\O$.
\end{theorem}
\begin{proof}
Suppose we are given an MSO transduction, presented in the form with regular expressions instead of formulas. Given an input string, a position $x$ in the string and a copy element $c$, in order to decide whether the vertex $(c,x)$ of the output DAG is labeled with $f$ we have to check if the suffix (from $x$ to the end) satisfies the regular expression $r_2$, where $\phi^c_f = r_1\sepa r_2$ is a vertex rule of the transduction. Similarly, to decide whether an edge $c \to_i d$ emanates from copy $c$ at position $x$ we have to check if the suffix (from $x$ to end) satisfies the expressions $r_2 \cdot r_3$, where $\psi^{c,d}_i = r_1\sepa r_2\sepa r_3$ is an edge rule of the transduction. Finally, to decide whether an edge $d \to_i c$ has the vertex $(c,x)$ as destination we have to check if the suffix (from $x$ to end) satisfies the expression $r_3$, where $\psi^{d,c}_i = r_1\sepa r_2\sepa r_3$ is an edge rule of the transduction.

For a regular expression $r$ over $\Sigma$, we write $\sem{r} \subseteq \Sigma\kstar$ to mean the regular set that $r$ denotes. The \emph{derivative} of a language $L \subseteq \Sigma\kstar$ w.r.t.\ a word $u \in \Sigma\kstar$ is the language $D_u(L) = \{ v \in \Sigma\kstar \mid uv \in L \}$. Observe that $D_v(D_u(L)) = D_{uv}(L)$ for all $L \subseteq \Sigma\kstar$ and $u,v \in \Sigma\kstar$.

Let $\R$ be the collection of all regular sets needed for these suffix tests, that is, $\sem{r_2}$ from every vertex rule $r_1\sepa r_2$, as well as both $\sem{r_2 \cdot r_3}$ and $\sem{r_3}$ from every edge rule $r_1\sepa r_2\sepa r_3$. Now, define $\R'$ to be the smallest collection of regular sets that contains $\R$ and is closed under derivatives. It can be seen that $\R' = \tbigcup_{L \in \R} \R'_L$, where $\R'_L$ is the collection of regular sets obtained in the following way: take the collection of all derivatives of $L$, which correspond to the unique minimal automaton for $L$. Finally, define $\At_\R$ to be the collection of atoms (nonempty minimal elements) of the Boolean algebra generated by $\R'$. The elements of $\At_\R$ partition $\Sigma\kstar$ (they are pairwise disjoint and their union is equal to $\Sigma\kstar$), and they define an \emph{unambiguous NFA}, where the language accepted from every state $T \in \At_\R$ is equal to $T$. To see how the transitions of the NFA $\At_\R$ are calculated, consider a state (atom) $T = L_1 \cap \compl L_2 \cap L_3$ where each of $L_i$ is an element of $\R'$. For a letter $\tg{a} \in \Sigma$ the derivative of $T$ is
\begin{align*}
D_\tg{a}(T) &=
D_\tg{a}(L_1) \cap D_\tg{a}(\compl L_2) \cap D_\tg{a}(L_3)
\\ &=
D_\tg{a}(L_1) \cap \compl D_\tg{a}(L_2) \cap D_\tg{a}(L_3),
\end{align*}
and it is partitioned by some atoms $T'_1, \ldots, T'_\ell$ of $\At_\R$ because every $D_\tg{a}(L_i)$ belongs to $\R'$. So, $T \to^\tg{a} T'_i$ are all the transitions from $T$ on letter $\tg{a}$ in the unambiguous NFA $\At_\R$. The automaton $\At_\R$ thus constructed is called the \emph{future automaton} for the transduction, because every state specifies the possible future suffixes.

Similarly, the application of the rules of the MSO transduction depends on regular properties of the input prefix seen so far. For a vertex rule $r_1\sepa r_2$, an $\eps$-transition rule $r_1\sepa \eps\sepa r_2$, and a letter transition rule $r_1\sepa \tg{a}\sepa r_2$, we need to check whether the input prefix satisfies the pattern $r_1$. To perform all these tests we construct the \emph{past automaton}, which simulates the parallel execution of automata for all these prefix patterns. The past automaton is always a DFA.

The product of the future automaton and the past automaton is an unambiguous NFA, which we call the \emph{future-past automaton} where every state specifies exactly which rules of the MSO transduction apply at the current position. So, every state of the future-past automaton specifies precisely the \emph{shape} of a position: which copy elements are active, which $\eps$-edges are enabled, and which outgoing letter-edges are enabled. A state only needs to specify values for variables whose values are going to be used in the next step of the computation. The expressions for these outgoing variables can be read off immediately from the shape.
For example, suppose the state $T'$ of the future-past automaton specifies the following active vertices and edges:
\[
  \includegraphics{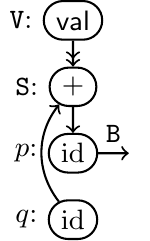}
\]
We know that any incoming transition $T \to T'$ from some state $T$ (e.g., states $q_1$, $q_3$ and $q_6$ in Figure~\ref{fig:futurePastAutShapes}) must specify a value for the variable $x_q$ (for the copy $q$), because it has to be used. The state $T'$ has to specify a new value for $x_p$, because the only outgoing letter-edge points to the copy $p$. So, the variable update for the transition $T \to T'$ must be $x_p := x_q + \val$ (see edges in Figure~\ref{fig:CRA1}). If a state is initial, then it must be possible to calculate the values for the variables without relying on any incoming value from a previous state (e.g., states $q_1$, $q_8$ and $q_9$ in Figure~\ref{fig:futurePastAutShapes}).
If the shape has an active copy without an outgoing edge, then its value is a potential final value for the computation (e.g., copy $\tg S$ in states $q_4$ and $q_9$ of Figure~\ref{fig:futurePastAutShapes}). For the particular case where no vertex is active (e.g., state $q_7$ in Figure~\ref{fig:futurePastAutShapes}), we have to propagate the final value of the computation (the variable where it is stored is uniquely determined).

It is easy to check that if the transduction is word-to-tree (that is, every vertex in a shape has at most one outgoing edge), then the constructed CRA is copyless.
\end{proof}

\begin{example}[Future Automaton]
\label{ex:futureAut}
\normalfont
For the single-step MSO transduction of Example~\ref{ex:transduction1singleStep} we will construct the future automaton, as described in the proof of Theorem~\ref{thm:MSOtoCRA}. The collection of regular sets needed for suffix tests is:
\[
  \R = \{
    \Sigma\kstar,\
    \tg{B}\kstar\tg{A}\Sigma\kstar,\
    \tg{B}\kplus\tg{A}\Sigma\kstar,\
    \tg{A}\Sigma\kstar
  \}.
\]
After closing under derivatives we obtain the collection
\[
  \R' = \{
    \Sigma\kstar,\
    \tg{B}\kstar\tg{A}\Sigma\kstar,\
    \tg{B}\kplus\tg{A}\Sigma\kstar,\
    \tg{A}\Sigma\kstar,\
    \emptyset
  \}.
\]
The set $\tg{B}\kstar\tg{A}\Sigma\kstar$ consists of all words that contain at least one $\tg{A}$, so its complement is $\tg{B}\kstar$. The set $\tg{B}\kplus\tg{A}\Sigma\kstar$ consists of all words that start with $\tg{B}$ and contain at least one $\tg{A}$, so its complement is $\tg{A}\Sigma\kstar \cup \tg{B}\kstar$. The complement of $\tg{A}\Sigma\kstar$ is $\eps \cup \tg{B}\Sigma\kstar$. The elements of $\At_\R$ are therefore the minimal nonempty intersections of elements from
$\tg{B}\kstar\tg{A}\Sigma\kstar$,
$\tg{B}\kplus\tg{A}\Sigma\kstar$,
$\tg{A}\Sigma\kstar$,
$\tg{B}\kstar$, and
$\eps \cup \tg{B}\Sigma\kstar$.
We claim that the Boolean atoms generated by this list are
\[
  \At_\R = \{
    \tg{A}\Sigma\kstar,\
    \tg{B}\kplus\tg{A}\Sigma\kstar,\
    \tg{B}\kstar
  \}.
\]
Indeed, $\At_\R$ partitions $\Sigma\kstar$ and moreover we have:
\begin{align*}
\tg{B}\kstar\tg{A}\Sigma\kstar &=
  \tg{A}\Sigma\kstar \cup
  \tg{B}\kplus\tg{A}\Sigma\kstar
\\
\eps \cup \tg{B}\Sigma\kstar &=
\tg{B}\kstar \cup \tg{B}\kplus\tg{A}\Sigma\kstar
\end{align*}
The unambiguous NFA $\At_\R$ is illustrated below:
\[
  \includegraphics{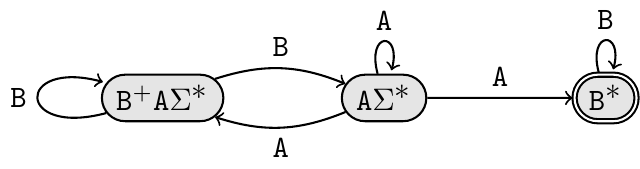}
\]
All the states are marked as initial (with the gray background) because the domain of the transduction is $\Sigma\kstar$.
\end{example}

\begin{example}[Past Automaton]
\label{ex:pastAut}
\normalfont
The expression $(\eps \cup \Sigma\kstar\tg{A})\tg{B}\kstar$ is equivalent to $\Sigma\kstar$, hence the rules of the single-step transduction of Example~\ref{ex:transduction1singleStep} can be simplified:
\begin{align*}
\begin{aligned}[t]
\phi^\tg{V}_\val &=
\Sigma\kstar \tg{A}\sepa \Sigma\kstar
\\
\phi^\tg{S}_0 &= \eps\sepa \Sigma\kstar
\\
\phi^\tg{S}_+ &=
\Sigma\kstar \tg{A}\sepa \Sigma\kstar
\\
\phi^{p}_\id &=
\Sigma\kstar\sepa
\tg{B}\kstar\tg{A}\Sigma\kstar
\\
\phi^{q}_\id &=
\Sigma\kstar\tg{A}\sepa \Sigma\kstar
\end{aligned}
&&
\begin{aligned}[t]
\psi^{\tg V, \tg S}_2 &=
\Sigma\kstar \tg{A}\sepa \eps\sepa \Sigma\kstar
\\
\psi^{\tg S, p}_1 &=
(\eps \cup \Sigma\kstar\tg{A})\sepa
\eps\sepa
\tg{B}\kstar\tg{A}\Sigma\kstar
\\
\psi^{p, p}_1 &=
\Sigma\kstar\sepa
\tg{B}\sepa
\tg{B}\kstar\tg{A}\Sigma\kstar
\\
\psi^{p, q}_1 &=
\Sigma\kstar\sepa
\tg{A}\sepa
\Sigma\kstar
\\
\psi^{q, \tg S}_1 &=
\Sigma\kstar\tg{A}\sepa
\eps\sepa
\Sigma\kstar
\end{aligned}
\end{align*}
The regular expressions $\Sigma\kstar\tg{A}$, $\eps$, $\Sigma\kstar$, $\eps \cup \Sigma\kstar\tg{A}$ are relevant for checking whether the input prefix seen so far satisfies the left-hand expression of a rule. The automaton that simulates the parallel execution of DFAs for these expressions is the following:
\[
  \includegraphics{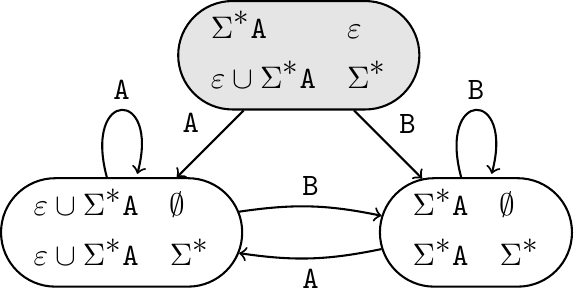}
\]
The labels are computed by taking derivatives of the expressions. It is more convenient to label a state $q$ with the regular expression that denotes the set of strings that lead from the initial state to $q$.
\[
  \includegraphics{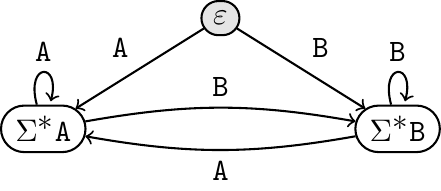}
\]
The label of each state determines if the automaton accepts for each one of the ``past'' expressions $\Sigma\kstar\tg{A}$, $\eps$, $\Sigma\kstar$, $\eps \cup \Sigma\kstar\tg{A}$.
\end{example}

\begin{example}[Future-Past Automaton, CRA]
\label{ex:futurePastAut}
\normalfont
The product of the future automaton from Example~\ref{ex:futureAut} and the past automaton from Example~\ref{ex:pastAut} is the \emph{future-past} automaton of Figure~\ref{fig:futurePastAut}. Every state $q$ is labeled with $L_q$ (top), the set of strings that lead from some initial state to $q$, and with $R_q$ (bottom), the set of strings that lead from $q$ to some final state. The purpose of these labelings is that they specify exactly which rules of the transaction are enabled at a particular state. In Figure~\ref{fig:futurePastAutShapes} every state of the future-past automaton is annotated with a shape that indicates the active copies and the edges ($\eps$-edges and letter-edges) that are enabled. From these shapes we can immediately read off the variable updates for the transitions, and the initialization of the variables in the initial states. The resulting CRA is shown in Figure~\ref{fig:CRA1}. Every state $q_i$ of the CRA is annotated with the set of variables that have to be set when the computation reaches $q_i$.
\end{example}

\begin{figure}
\centering
\includegraphics{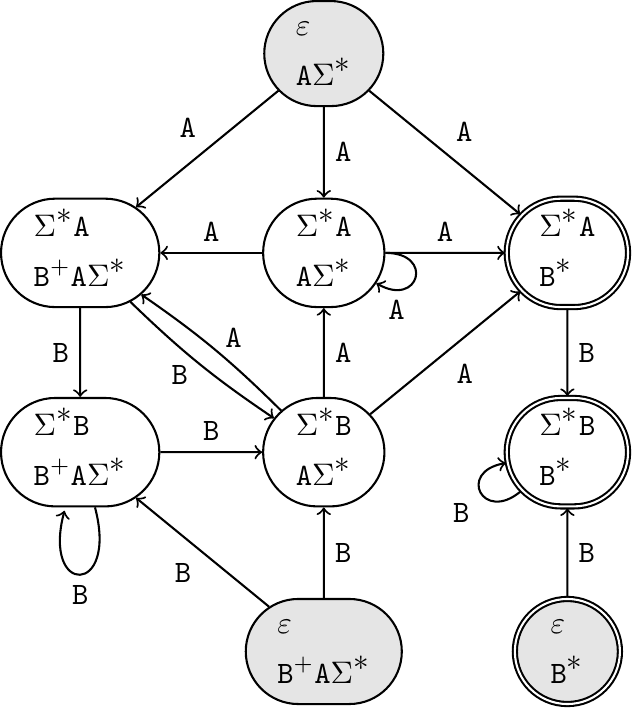}
\caption{The future-past automaton for the single-step transduction of Example~\ref{ex:transduction1singleStep}.}
\label{fig:futurePastAut}
\end{figure}

\begin{figure}
\centering
\includegraphics{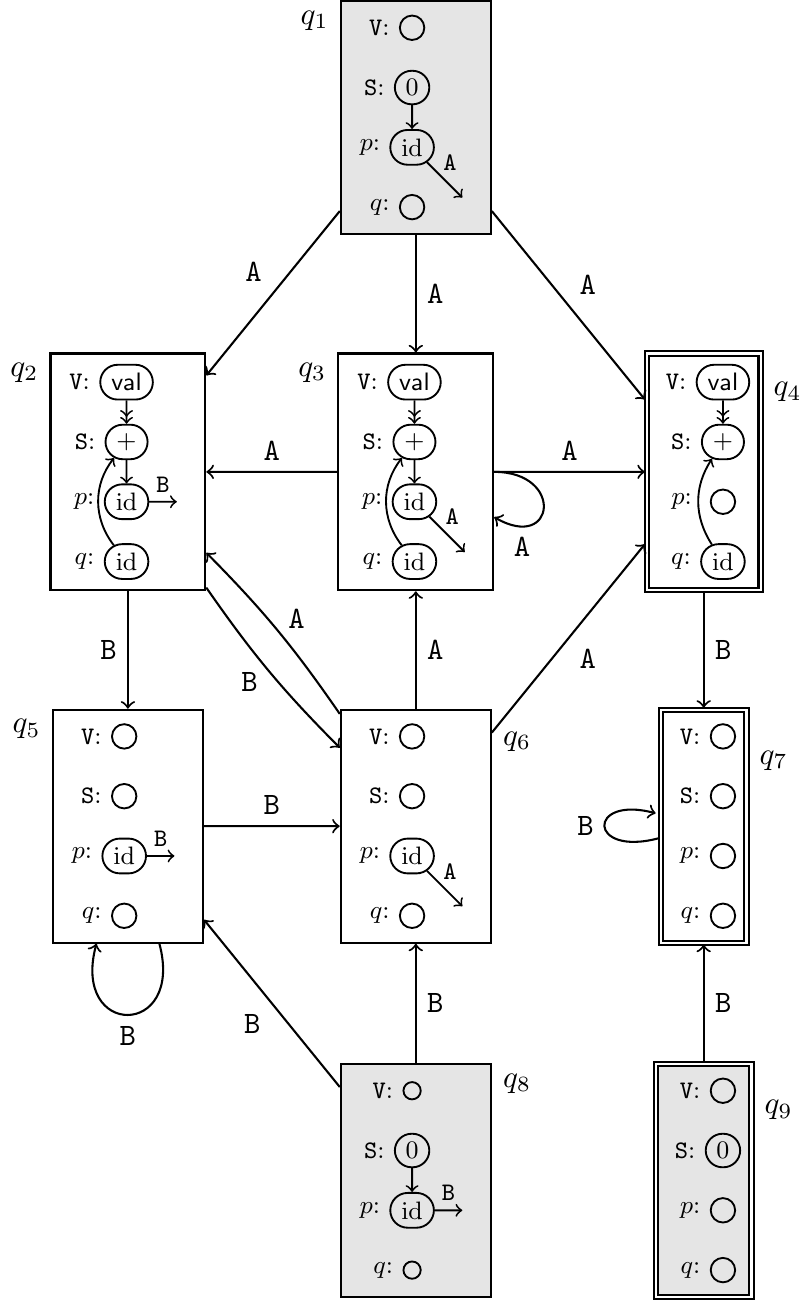}
\caption{The future-past automaton for the single-step transduction of Example~\ref{ex:transduction1singleStep}, annotated with ``shapes'' that show the enabled rules.}
\label{fig:futurePastAutShapes}
\end{figure}

\begin{figure}[t]
\centering
\includegraphics{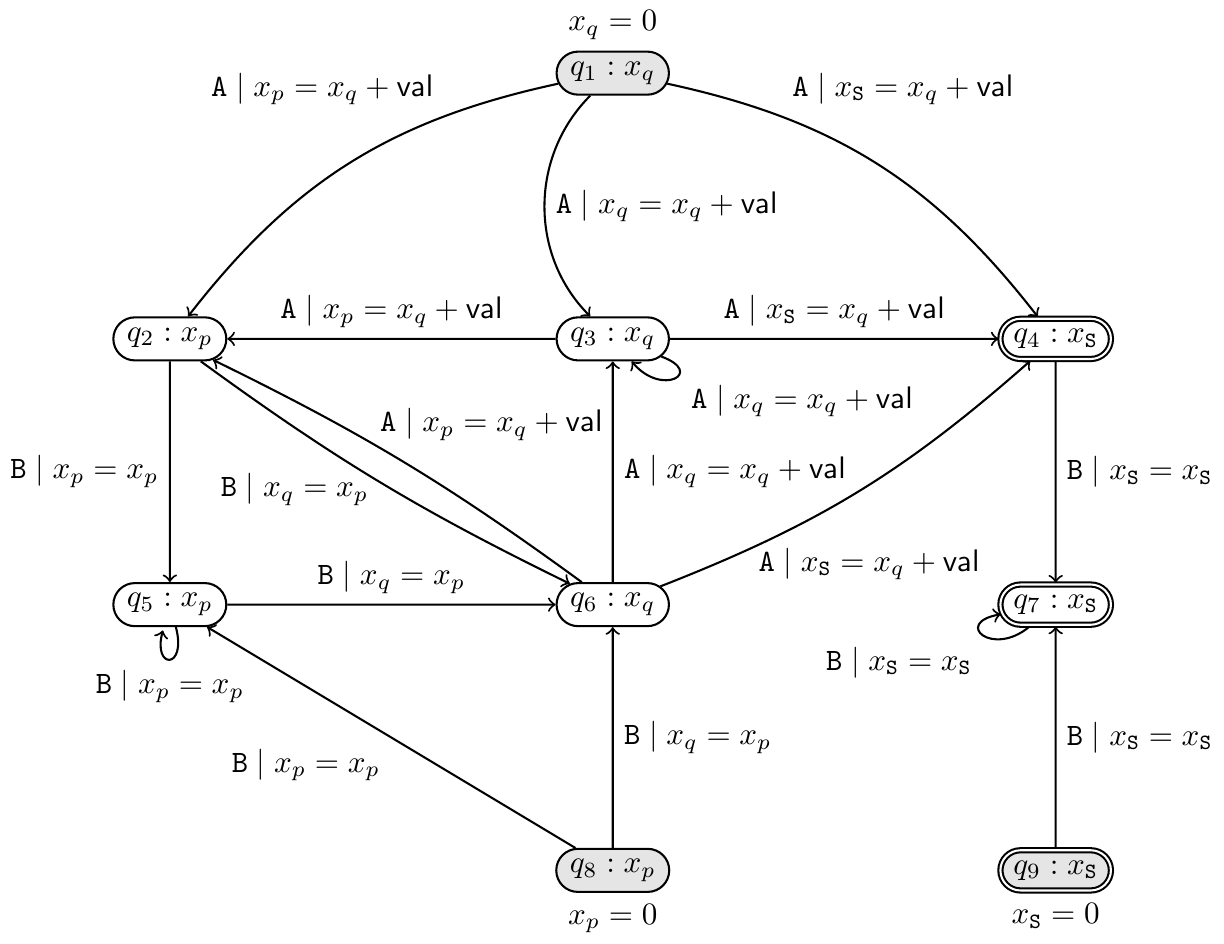}
\caption{Unambiguous CRA for the single-step transduction of Example~\ref{ex:transduction1singleStep}.}
\label{fig:CRA1}
\end{figure}

\begin{corollary}
\label{coro:MSOtoCRA}
\normalfont
For any $\O$, every forward-only word-to-DAG (resp., word-to-tree) MSO transduction over $\O$ can be implemented by a copyful (resp., copyless) UCRA over $\O$.
\end{corollary}
\begin{proof}
Suppose we are given an MSO transduction. By virtue of Lemma~\ref{lemma:MSOtoRegex} we can present it equivalently using regular expressions instead of MSO formulas. Lemma~\ref{lemma:singleStep} says that the transduction can be transformed into an equivalent one, with potentially more copy elements, in which edges are either $\eps$-edges or single-letter edges. From this modified transduction we can construct an unambiguous CRA that computes the transduction, as shown in Theorem~\ref{thm:MSOtoCRA}.
\end{proof}

\begin{theorem}
\label{thm:CRAtoMSO}
\normalfont
For every $\O$, the transductions of the class $\SR(\O)$ (resp., $\SLR(\O)$) can be defined in forward-only word-to-DAG (resp., word-to-tree) MSO.
\end{theorem}
\begin{proof}
With a construction that is a simple variant of the classical transformation of an automaton into a MSO formula defining the same language, we can convert an unambiguous CRA into a single-step MSO transduction. We need one copy element per register, as well as sufficient copies to construct the algebraic terms that appear in all variable updates of the CRA. The resulting MSO transduction may specify vertices labeled with the identity function $\id$. Sequences of edges of the form $f \to_1 \id \to_1 \cdots \to_1 \id \to_i g$ can be specified with an MSO formula, and can therefore be replaced by a single edge of the form $f \to_i g$. This results in a (not necessarily single-step) MSO transduction that specifies the desired function. This construction handles copyful and copyless UCRAs uniformly: i.e., on a copyless UCRA it produces a word-to-tree MSO transduction. The well-formedness conditions for MSO transductions hold, including that all edges are forward-only.
\end{proof}

\section{Relation to Weighted Automata}
\label{sec:WA}

In this section we investigate the relationship between weighted automata and CRAs. We will show, in particular, that there are CRAs for suitable choice of $\O$ that correspond precisely to unambiguous and nondeterministic weighted automata respectively. This means that CRAs are a more general model of quantitative automata. In section~\ref{sec:related} we give a brief overview of relevant literature from the line of work on weighted automata.

The model of weighted automata is a widely studied quantitative extension of nondeterministic finite-state automata \cite{DKV2009HWA} that describe (total) functions from $\Sigma\kstar$ to a set of data values $D$. The set $D$ is typically equipped with two operations, product and sum, that form a semiring. In a weighted automaton every transition is labeled with a tag in $\Sigma$ and a weight in $D$. The weight of a run of the automaton is the product of all transition weights in the run. The output weight associated with a input string $w$ is the sum of the weights of all the runs that are labeled with $w$. Unlike CRAs, weighted automata define total functions: if there are no such runs, the output weight would be the zero of the semiring.

Formally, a \emph{semiring} (also called a \emph{rig}) is an algebraic structure $\sring D = (D,+,\cdot,0,1)$ where $+$ is a commutative, associative binary operation on $D$ with identity $0$, $\cdot$ is an associative binary operation on $D$ with identity $1$ and absorbing element $0$, and $\cdot$ distributes over $+$.
A \emph{weighted automaton} over a semiring $\sring D$ is a tuple $\aut A = (Q,\Sigma,\Delta,I,F)$,
where $Q$ is the set of states, $\Sigma$ is the input alphabet, $\Delta: Q \times \Sigma \times Q \to D$ is the transition function, $I: Q \to D$ is the initial weight function, and $F: Q \to D$ is the final weight function. For a word $u = a_1 a_2 \cdots a_n \in \Sigma\kstar$, a \emph{$u$-path} in $\aut A$ is a sequence
\[
  \pi =
  q_0
  \xlongrightarrow{a_1/d_1} q_1
  \xlongrightarrow{a_2/d_2} q_2
  \cdots
  \xlongrightarrow{a_n/d_n} q_n,
\]
where $\Delta(q_{i-1},a_i,q_i) = d_i$ for all $i = 1, 2, \ldots, n$. The \emph{weight} of the path $\pi$ is $\weight(\pi) = d_1 \cdot d_2 \cdots d_n$. We denote by $\first(\pi)$ and $\last(\pi)$ the first state $q_0$ and the last state $q_n$ of the path $\pi$ respectively.
Finally, the \emph{weight} of a word $u \in \Sigma\kstar$ in $\aut A$ is given by
\[
  \weight(u) =
  \tsum_\text{$u$-path $\pi$}\,
    I(\first(\pi)) \cdot \weight(\pi) \cdot F(\last(\pi)).
\]
The weighted automaton $\aut A$ computes the function $\sem{\aut A}: \Sigma\kstar \to D$, where $\sem{\aut A}(u)$ is the weight of $u$ in $\aut A$.

The set of \emph{initial states} of $\aut A$ is $\{ q \in Q \mid I(q) \neq 0 \}$, and similarly the set of \emph{final states} of $\aut A$ is $\{ q \in Q \mid F(q) \neq 0 \}$. A path $\pi$ in $\aut A$ is \emph{successful} if $\first(\pi)$ is initial and $\last(\pi)$ is final. A path is \emph{unsuccessful} if it is not successful. Observe that the weight of an unsuccessful path is equal to 0. A weighted automaton $\aut A$ is said to be \emph{unambiguous} if for every word $u \in \Sigma\kstar$ there is at most one successful $u$-path in $\aut A$.

An unambiguous weighted automaton can be over a monoid, instead of a semiring: modify $I$ and $F$ to be partial functions instead of total, where their domains are the initial and final states, respectively. Then removing $+$ and $0$, we are left a single associative operation $\cdot$ on $D$ with identity $1$. The unambiguous automaton then defines, in general, a \emph{partial} function from input strings to elements of the monoid $(D, \cdot, 1)$. The following theorem shows a correspondence between unambiguous weighted automata over the monoid $(D,\cdot,1)$ and \emph{copyless} UCRAs that can use the multiplication operation $\cdot$ in a certain restricted way: only multiplying by constants on the right. For a constant $d \in D$, we write $(- \cdot d)$ for the function that multiplies by $d$ on the right.

Recall from Section~\ref{sec:choiceOps} that, for a monoid $M$, the set of operations $\O^\unary_M$ contains the identity $1$ and the unary operation $(- \cdot d)$ for every $d \in D$.

\begin{theorem}
\label{thm:UWA}
\normalfont
Suppose that $M = (D,\cdot,1)$ is a monoid, i.e.\ $\cdot$ is an associative binary operation on $D$ and $1$ is a left and right identity for $\cdot$. Then the following are equal: (1) the class of total functions $\Sigma\kstar \to D$ computed by unambiguous weighted automata over $M$, (2) the class of total transductions in $\SLR(\O^\unary_M)$, (3) the class of total transductions in $\SR(\O^\unary_M)$.
\end{theorem}
\begin{proof}
Recall from Lemma~\ref{lemma:monoidWeak} that $\SR(\O^\unary_M) = \SLR(\O^\unary_M)$. So we just have to show that unambiguous weighted automata over $(D, \cdot, 1)$ are expressively equivalent to copyless UCRAs over $\O^\unary_M$.

It is easy to see that every unambiguous weighted automaton over $(D,\cdot,1)$ can be simulated by a copyless single-register UCRA over $\O^\unary_M$. So it remains to show that the computation of a copyless UCRA $\aut A$ over $\O^\unary_M$ can be simulated by an unambiguous weighted automaton $\aut B$ over $(D,\cdot,1)$. We apply the same construction as in Lemma~\ref{lemma:monoidWeak}, based on the idea that at every step of the computation of $\aut A$, the content of at most one register can contribute to the final output. So, $\aut B$ is defined to have state space $Q \times (X \cup \{\bot\})$, where $Q$ is the state space of $\aut A$ and $X$ is the set of registers of $\aut A$. We also may need additional states to compute the finalization function.
\end{proof}



Similarly to Theorem~\ref{thm:UWA} we now establish a correspondence between (potentially ambiguous) weighted automata over a semiring $(D,+,\cdot,0,1)$ and copyful CRAs that can use the multiplication operation $\cdot$ only to right-multiply by constants.

\begin{theorem}
\normalfont
\label{thm:RtoSR}
Suppose that $D$ is the type of data values and $\sring D = (D,+,\cdot,0,1)$ is a semiring. Moreover, assume that $\O$ consists of the constants $0, 1$, the binary operation $+$, and the family of unary operations $(- \cdot d)$ for every $d \in D$. The class of total transductions of $\SR(\O)$ is equal to the class of transductions computed by weighted automata over $\sring D$.
\end{theorem}
\begin{proof}
Every weighted automaton $(Q,\Sigma,\Delta,I,F)$ can be simulated by a copyful single-state CRA by considering one register $x_q$ for each state $q$ of the weighted automaton. The registers are used to encode the configuration of the weighted automaton, which is a mapping from the state space to $D$. For every tag $a$, we include an $a$-labeled transition in the CRA that updates the variable $x_q$ with the assignment $x_q := \tsum_{p \in Q}\ x_p \cdot \Delta(p,a,q)$.

Consider now a copyful CRA $\aut A$ that computes a total transduction. It remains to show that the computation of $\aut A$ can be simulated by a weighted automaton. We will describe a weighted automaton $\aut B$ that computes the same transduction as $\aut A$. We may algebraically simplify every update of $\aut A$ to be of the form
\[
  x = x_1 \cdot d_1 + \cdots + x_n \cdot d_n + d,
\]
where $x_1, \ldots, x_n$ is an enumeration of the registers of $\aut A$ and $d, d_i \in D$ are constant weights. The state space of $\aut B$ consists of states $q$ and pairs $(q,x)$, where $q$ is a state of $\aut A$ and $x$ is a register of $\aut A$. Using the above algebraic simplification for each register $x_1$ to $x_n$, the register update function of a transition $q \to^a q'$ of $\aut A$ can be presented in the following form:
\[
  \begin{bmatrix}
    x_1 \\ x_2 \\ \vdots \\ x_n
  \end{bmatrix}^\top
  =
  \begin{bmatrix}
    x_1 \\ x_2 \\ \vdots \\ x_n
  \end{bmatrix}^\top
  \cdot
  \begin{bmatrix}
    d_{11} & d_{12} & \cdots & d_{1n} \\
    d_{21} & d_{22} & \cdots & d_{2n} \\
    \vdots & \vdots & \ddots & \vdots \\
    d_{n1} & d_{n2} & \cdots & d_{nn}
  \end{bmatrix}
  +
  \begin{bmatrix}
    d_1 \\ d_2 \\ \vdots \\ d_n
  \end{bmatrix}^\top,
\]
where $d_{i,j} \in D$ and $d_i \in D$ are constant weights.
The above matrix equation can be reformulated as follows:
\[
  x_j = (\tsum_{i=1}^n x_i \cdot d_{ij}) + d_j
  \ \text{for all $j=1,\ldots,n$}.
\]
The update function is simulated in the weighted automaton $\aut B$ by defining the following transitions:
\begin{align*}
\Delta((q,x_i),a,(q',x_j)) &= d_{ij}
&
\Delta(q,a,(q',x_j)) &= d_j
&
\Delta(q,a,q') &= 1
\end{align*}
The function $\Delta$ assigns the weight $0$ to every other possible transition. If $q$ is a final state of $\aut A$ with finalization function $x_1 \cdot d_1 + \cdots + x_n \cdot d_n + d$, we define the final weight function $F$ of $\aut B$ by $F(q) = d$ and $F((q,x_i)) = d_i$. If $q$ is not a final state, then we put $F(q) = 0$ and $F(q,x) = 0$.
Suppose that $I$ is the initialization function of $\aut A$. We will describe the initial weight function $J$ of $\aut B$. If $q$ is an initial state of $\aut A$, we define $J(q) = 1$ and $J((q,x)) = \sem{I(q)(x)}$. If $q$ is not an initial state of $\aut A$, we define $J(q) = 0$ and $J((q,x)) = 0$.
\end{proof}






It is important that we have restricted the allowed operations of the CRA to include only right-multiplication by a constant, a unary operation. If a (copyful) CRA can use the binary $\cdot$ operation of the semiring $(D,+,\cdot,0,1)$, then it can compute the function $a^i \mapsto d^{2^i}$ (where $a \in \Sigma$ and $d \in D$), which is not implementable by a weighted automaton.

\section{Related Work}
\label{sec:related}

\paragraph*{\bf\em Weighted automata and extensions}

Weighted Automata, which were introduced in 1961 by Sch\"utzenberger \cite{S1961WA} (see also the more recent monograph \cite{DKV2009HWA}), extend classical nondeterministic automata by annotating transitions with \emph{weights} and can be used for the computation of simple quantitative properties on finite or infinite strings of symbols \cite{CDH2010QL}. Weighted automata have found applications in speech and language processing \cite{M1997FST}, and they have been extended for modeling systems and verifying quantitative properties of these systems \cite{CHO2016QMA}. The studied models for the latter purpose include \emph{Nested Weighted Automata} \cite{CHO2015NWA}, a two-level variant of weighted automata for infinite strings. The computational problems that are relevant for quantitative verification are analysis questions such as universality and equivalence. These questions are decidable only when the weights and the operations used on them are very simple \cite{K1994EPRS,ABK2011WA}, so the studied models are usually equipped with a very limited set of primitive operations that are insufficient for expressing realistic computations over streaming data. We have shown in Section~\ref{sec:WA} that CRAs generalize weighted automata. In particular, when the family $\O$ of available operations is chosen appropriately, CRAs correspond precisely to unambiguous and nondeterministic weighted automata respectively.


\paragraph{\bf\em Register automata}

Another approach to augment classical automata with quantitative features has been with the addition of \emph{registers} that can store values from a potentially infinite set. These models are typically varied in two aspects: by the choice of data types and operations that are allowed for register manipulation, and by the ability to perform tests on the registers for control flow. 

The literature on data words, data/register automata and their associated logics \cite{KF1994FMA, NSV2004FSM, DL2009LFQ, BS2010NRDL, BDMSS2011LDW} studies models that operate on words over an infinite alphabet, which is typically of the form $\Sigma \times \mathbb{N}$, where $\Sigma$ is a finite set of tags and $\mathbb{N}$ is the set of the natural numbers. They allow comparing data values for equality, and these equality tests can affect the control flow. Due to this feature many interesting decision problems become undecidable \cite{NSV2004FSM}.

\paragraph{\bf\em Cost register automata}
Cost Register Automata~(CRAs)~\cite{AdADRY2013CRA, AR2013ARF} and Streaming String Transducers~(SSTs)~\cite{AC2010SST,
AC2011STA, AdA2012STT} have recently been proposed as models to permit more complex operations over the cost domain.
Unlike the case of register automata, which traditionally only allow checking for equality, CRAs are parameterized by
the set of permissible operations $\O$ over the values of the registers. On the other hand, CRAs and SSTs also enforce a
strict separation between data and control, so that data values from the input sequence can only affect the contents of
the registers but cannot determine the state transitions. SSTs and CRAs (we are referring to the original deterministic CRA definition of \cite{AdADRY2013CRA} with demands copylessness and allows terms with parameters in the registers) have robust expressive power~\cite{AC2010SST,
AdADRY2013CRA}, including equivalent characterizations in MSO  and closure under various transformations, and have been
used in applications such as verifying list processing programs~\cite{AC2011STA}. Work has been devoted to study the
circuit complexity of the functions being computed~\cite{AM2015, AKM2017}, and on analysis problems such as determining
the register complexity~\cite{AR2013ARF, DRT2016}. Previous work on CRAs has either focused on some concrete families of allowed operations (e.g. addition and multiplication over a semiring~\cite{MR2015}), or it has allowed registers whose contents may include \emph{syntactic terms} rather than concrete values from the data domain. In this work, we study CRAs over a general set of operations, either copyful or copyless, and we do not allow syntactic terms, so as to identify the \emph{streamable} regular functions.

More recently,~\cite{AMS2019} proposes a model that is expressively equivalent to CRAs, but exponentially more succinct, and still permits streaming evaluation. Rather than having a finite set of states and a separate finite set of registers, this model (Data Transducers) employs a finite set of combined ``state variables'' which are either inactive or active (like a CRA state), but also hold a value if they are active (like a CRA register). A more succinct model than CRAs is appropriate for the design of practical query languages based on streamable regular functions (see the paragraph on quantitative regular expressions, below).

\paragraph{\bf\em MSO-definable graph transformations}
An important goal for language theorists has been to characterize operational models such as automata and transducers
using properties and functions expressible in various logical theories. The earliest of these results includes the
seminal work of B\"uchi~\cite{B60} and Elgot~\cite{E61}, establishing the equivalence between regular languages and
properties expressible using monadic second-order logic (MSO). Courcelle considered the notion of MSO-definable graph
transductions~\cite{C1994MSOGT}. Engelfriet and Hoogeboom~\cite{EH2001MDST} showed the equivalence of string-to-string
transformations expressible in Courcelle's framework with those computable using two-way finite state transducers, and Alur and
\v{C}erny equated their expressive power with the deterministic one-way model of streaming string transducers~(SSTs)~%
\cite{AC2010SST}. These have also been extended to transformations of infinite strings~\cite{AFT2012RTIS}. Streaming
tree transducers~\cite{AdA2012STT} and CRAs~\cite{AdADRY2013CRA} are similar characterizations of tree-to-tree and
string-to-tree transformations respectively. With this background, the present paper shows that copyful CRAs over $\O$
can compute exactly those functions which can be expressed as word-to-DAG transformations in MSO, i.e. the class of
streamable regular transductions, $\SR(\O)$.

\paragraph{\bf\em Quantitative regular expressions}
The connection between finite automata and regular expressions has motivated research into identifying similar
linguistic characterizations of the transductions computable by CRAs and SSTs~\cite{AFR2016QRE, AFR2014LICS}. The resulting formalisms, \emph{Quantitative Regular Expressions}~(QREs) and DReX respectively, have subsequently been applied in expressing string transformations~\cite{AdAR2015DReX}, monitoring network traffic~(NetQRE)~\cite{YLMMAL2017NQRE}, and processing streaming data~(StreamQRE)~\cite{MRAIK2017SQRE, AM2017SQRE}. StreamQRE has been used to specify streaming algorithms for medical monitoring \cite{ARMBSG2018, AAMMR2018}, and the efficiency of its evaluation has been investigated in \cite{AMS2019, AMS2017, AMU2017}. StreamQRE has also inspired work on the design of novel type-based language abstractions for distributed stream processing \cite{MSAIT2019, AMST2018}.
The closure of the transduction classes $\SR$ and $\SLR$ under
the regular combinators naturally facilitates the design of these
declarative languages.
These papers are then largely concerned with twin
problems of expressiveness---providing translation algorithms to and from CRAs and SSTs---and fast query evaluation. As
with previous research on CRAs, they focus on copyless models, and exploit these for time- and memory-efficient one-pass query evaluation.
An important technical challenge is the design of
a type system which constrains the rate of the transduction and the
composition rules to permit efficient modular evaluation.
The combinators of Table~\ref{table:combinators}, for example, have been refined in \cite{MRAIK2017SQRE, AM2017SQRE} with a type system that eliminates sources of complexity.
When the registers of CRAs maintain terms, an important consideration is term \emph{simplification}, so that the
contents of each register may be maintained in constant space, regardless of the length of the input stream seen so far%
~\cite{AFR2016QRE}. In this context, developing a regular expression-like characterization of the class of streamable
regular transductions, $\SR(\O)$, which we consider in this paper, is a problem for future work.

\section{Conclusion}
\label{sec:conclusion}

We have studied the class of \emph{streamable regular transductions} ($\SR$), which are partial functions from input streams of tagged values to output values that can be computed by unambiguous Cost Register Automata (UCRAs). The subclass of \emph{streamable linear regular transductions} ($\SLR$) consists of those transductions that are computed by \emph{copyless} UCRAs, where the register updates are restricted so that each register appears at most once in the right-hand side of a parallel update. We have shown that the classes $\SR$ and $\SLR$ have appealing logical characterizations: $\SR$ (resp., $\SLR$) corresponds to MSO-definable transformations from strings to DAGs (trees) without backward edges. A precise relationship with the classical model of weighted automata has also been established.


\section*{Acknowledgements}

\noindent
This work was supported by NSF award CCF 1763514.

\section*{References}

\bibliographystyle{elsarticle-num} 
\bibliography{regular-functions}

\end{document}